\documentclass[journal]{IEEEtran}
\usepackage{cite}
\usepackage{fnpct}

\ifCLASSINFOpdf
\else
\fi
\usepackage{amsmath}
\allowdisplaybreaks
\usepackage{amsfonts}
\usepackage{mathtools}
\usepackage{amsthm}
\theoremstyle{definition}
\newtheorem{definition}{Definition}
\newtheorem{prop}{Proposition}
\newtheorem{theorem}{Theorem}

\theoremstyle{remark}
\newtheorem{remark}{Remark}
\usepackage[ruled,norelsize,linesnumbered]{algorithm2e}

\makeatletter
\newcommand{\removelatexerror}{\let\@latex@error\@gobble}
\makeatother
\usepackage{color}
\usepackage[normalem]{ulem}
\usepackage{tabularray}
\usepackage{booktabs}
\usepackage[flushleft]{threeparttable}
\usepackage[caption=false,font=footnotesize]{subfig}
\usepackage[utf8]{inputenc}
\usepackage[T1]{fontenc}
\usepackage{tikz,pgfplots}
\usetikzlibrary{positioning}
\usepackage{xpatch}
 
\makeatletter
\def\changeBibColor#1{%
  \in@{#1}{chen2023deep,fag2011on,mu2018comment,venzke2020inexact,liu2019acopf,mhanna2018adaptive,10.5555/2969442.2969497,brust2022con,jere2016sto,moustapha2022active,klink2022efficient,peng2021multi,glover2012power,kyri2019joint,church1951george, boole_2009,lof1992fast,bompard1996dynamic}%  list of colored bib items
  \ifin@\color{black}\else\normalcolor\fi
}

\xpatchcmd\@bibitem
{\item}
{\changeBibColor{#1}\item}
{}{\fail}
 
\xpatchcmd\@lbibitem
{\item}
{\changeBibColor{#2}\item}
{}{\fail}
\makeatother
\begin{document}

\title{Computationally Enhanced Approach for Chance-Constrained OPF Considering \\Voltage Stability}

\author{Yuanxi~Wu,~\IEEEmembership{Graduate~Student~Member,~IEEE,}
        Zhi~Wu,~\IEEEmembership{Senior~Member,~IEEE,}
        Yijun~Xu,~\IEEEmembership{Senior~Member,~IEEE,}
        Huan~Long,~\IEEEmembership{Member,~IEEE,}
        Wei~Gu,~\IEEEmembership{Senior~Member,~IEEE,}
        Shu~Zheng,
        and~Jingtao~Zhao% <-this % stops a space
        \thanks{This work has been submitted to the IEEE for possible publication. Copyright may be transferred without notice, after which this version may no longer be accessible.}
        \thanks{The study was supported by the science and technology project of the State Grid Corporation of China, Grant 5108-202218280A-2-366-XG.}
\thanks{Yuanxi Wu, Zhi Wu, Yijun Xu, Huan Long, Wei Gu, and Shu Zheng are with the School
of Electrical Engineering, Southeast University, Nanjing 210096, China (e-mail: yuanxi\underline{~~}wu@seu.edu.cn; zwu@seu.edu.cn; yijunxu@seu.edu.cn; hlong@seu.edu.cn; wgu@seu.edu.cn; zhengshu@sgepri.sgcc.com.cn).}% <-this % stops a space
\thanks{Jingtao Zhao is with State Grid Electric Power Research Institute, Nanjing 211102, China (e-mail: zhaojingtao@sgepri.sgcc.com.cn).}% <-this % stops a space
}

 \markboth{SUBMITTED FOR REVIEW}%
 {Shell \MakeLowercase{\textit{et al.}}: Bare Demo of IEEEtran.cls for IEEE Journals}

\maketitle

\begin{abstract}
The effective management of stochastic characteristics of renewable power generations is vital for ensuring the stable and secure operation of power systems. This paper addresses the task of optimizing the chance-constrained voltage-stability-constrained optimal power flow (CC-VSC-OPF) problem, which is hindered by the implicit voltage stability index and intractable chance constraints.
Leveraging a neural network (NN)-based surrogate model, the stability constraint is explicitly formulated and \textcolor{black}{directly} integrated into the model. To perform uncertainty propagation without relying on presumptions or complicated transformations, an advanced data-driven method known as adaptive polynomial chaos expansion (APCE) is developed. To extend the scalability of the proposed algorithm, a partial least squares (PLS)-NN framework is designed, which enables the establishment of a parsimonious surrogate model and efficient computation of large-scale Hessian matrices. In addition, a dimensionally decomposed APCE (DD-APCE) is proposed to alleviate the "curse of dimensionality" by restricting the interaction order among random variables. Finally, the above techniques are merged into an iterative scheme to update the operation point.
Simulation results reveal the cost-effective performances of the proposed method in several test systems. 
\end{abstract}

\begin{IEEEkeywords}
Optimal power flow, chance constraint, PLS, neural network, dimensionally decomposed APCE, data-driven.
\end{IEEEkeywords}

\IEEEpeerreviewmaketitle

\section{Introduction}\label{introduction}
\IEEEPARstart{A}{s} the world shifts up a gear in the development of renewable energy resources, uncertainties in power systems inevitably arise since renewable generations can hardly be forecasted with high accuracy \cite{wang2023two}.
Deterministic AC optimal power flow (OPF), although widely adopted in the traditional power industry to achieve reliable system operations, is not intended for such circumstances. Failing to incorporate stochastic features would result in significant degradation of system performance \cite{yang2023tracking}, which urges the necessity of moving beyond this classical setting.  

Different approaches exist to immunize power systems against uncertain injections. One intensively studied type is robust optimization \cite{qiu2020historical}, which enforces the power system to remain on the safe side for all possible uncertain-but-bounded perturbations. However, this worst-case-oriented philosophy is inclined to deliver an over-conservative solution, thus rendering itself less preferable. Much recent research effort has been directed to chance-constrained AC-OPF (CC-AC-OPF), which requires the uncertainty-affected constraints to be satisfied with a predefined probability. While capable of promoting operational efficiency, CC-AC-OPF is, more often than not, a heavily problematic task \cite{pena2020solving}. The crux of its solution lies in the \textit{ efficient evaluation of chance constraints}, which constitutes the first research concern of this paper. 

In light of this challenge, the traditional Monte-Carlo (MC) simulation offers a straightforward and conceptually clear-cut solution \cite{zhang2010probabilistic}. However, this method is extremely time-consuming, even for moderately sized problems. Du \textit{et al.} \cite{du2021chance} adopt the point estimation method to avoid massive sampling. Notwithstanding the improved efficiency, the accuracy of this method is not guaranteed, and its application to non-Gaussian random variables with nonlinear correlations is quite limited. 

Alternatively, the analytical method offers a viable approach for transforming chance constraints into tractable formulations. Early research relies on strong assumptions regarding characteristics of renewable generations, such as Gaussianity and independence. Instead, Yang \textit{et al.} \cite{yang2020analytical} embrace the adoption of the Gaussian mixture model to enable general characterizations of non-Gaussian uncertainties. Another widely adopted technique is the distributionally robust model, which utilizes certain statics of uncertainties \cite{roald2018chance,li2019distribution,arab2022distributionally,arrigo2022embed}. It is notable, however, that most of the studies cited above resort to approximations of AC power flow (e.g., linearized or relaxed power flow).

Recently, meta-modeling techniques have received significant attention due to their ability to inexpensively deal with uncertainty in AC power systems. \textcolor{black}{One choice is neural-network-based surrogate model. For example, Chen \textit{et al.} \cite{chen2023deep} propose a deep-quantile-regression-based model to tackle joint chance-constrained OPF. Nevertheless, it typically demands a substantial number of historical operational samples and a hyperparameter tuning process to achieve acceptable accuracy. This paper will focus on a more interpretable approach based on variants of polynomial chaos expansion (PCE).} Despite previous successful attempts to apply PCE to CC-OPF \cite{li2018compressive,muhlpfordt2019chance,metivier2020efficient,xu2021iterative}, a common trait of these works is that they have simplified the underlying dependence structure of random variables to varying degrees.
While recent studies in probabilistic power flow have made progress via the copula technique and Rosenblatt transformation
 \cite{ye2022generalized,ly2022scalable}, selecting an appropriate copula for modeling the joint probability distribution can prove difficult. Furthermore, the highly nonlinear Rosenblatt transformation risks compromising the performance of PCE in practical applications \cite{lee2022reliability}. A promising avenue for improving PCE is the direct construction of multivariate orthogonal polynomials with respect to the original probabilistic space. To the best of our knowledge, this approach is yet to receive sufficient exploration.

 In addition, voltage stability is crucial to ensure the robustness of OPF solutions against system instability. Traditionally, voltage-stability-constrained OPF (VSC-OPF) has been the preferred method to accomplish this goal \cite{jia2023voltage,cui2018new,avalos2008practical}. VSC-OPF enforces a predetermined threshold on the voltage stability index (VSI), typically represented by the minimum singular value of the power flow Jacobian matrix. In the context of high renewable penetration, the system state becomes more volatile, leading to an increased risk of static voltage instability.
 However, incorporating voltage stability constraints becomes more intricate when VSI is subjected to uncertainties. Selecting an appropriate operation point to ensure the probability distribution function of VSI meets the secure operation requirement is a challenging problem.
 Currently, there is very limited research on stability issues in CC-OPF problems, with only a few studies considering small-signal stability \cite{wang2023stability} or static voltage stability margin \cite{yang2022optimization}. 
 Hence, further investigation is necessary to advance the field of \textit{CC-VSC-OPF}, which constitutes the second research concern of this paper. 

In response to these pressing concerns, this study applies a novel adaptive polynomial chaos expansion (APCE) \cite{lee2020practical} to the CC-VSC-OPF problem, where a neural network (NN) regression model is adopted as a surrogate for the voltage stability constraint. The proposed APCE method facilitates stable computation of orthonormal polynomial basis even under arbitrarily distributed and complexly correlated uncertainties. These basis functions are ultimately integrated into an iterative framework to obtain the optimal operation point. Partial least squares (PLS) and a dimensionwise decomposition technique are further incorporated to improve the scalability of the proposed approach. The contributions of this study are summarized as follows:
\begin{itemize}
\item A novel CC-VSC-OPF model is designed with the consideration of voltage stability constraint and uncertain renewable generations, thereby maintaining stable and secure operations of power systems. 
\item \textcolor{black}{This study merges the NN-based surrogate model and APCE into an iterative scheme to deal with the implicit VSI and chance constraints without relying on simplified power flow models.}
\item \textcolor{black}{The scalability of the proposed algorithm is extended by introducing the PLS-NN framework. A dimensionally decomposed APCE (DD-APCE) is further devised to enable the implementation of the proposed algorithm under high-dimensional uncertainties.} 
\item \textcolor{black}{The proposed algorithm is applicable to general cases since it is entirely data-driven and distribution-free.}
\end{itemize}
The remainder of this paper is organized as follows: Section \ref{preliminary} introduces the formulation of CC-VSC-OPF. Section \ref{approach} presents the data-driven APCE and its dimensionally decomposed version. Section \ref{scheme} outlines the PLS-NN framework, along with the whole solution procedure for CC-VSC-OPF. Numerical experiments are conducted in Section \ref{simulation}, and conclusions follow in Section \ref{conclusion}. 
\section{Preliminaries}\label{preliminary}
This section first outlines the formulation of CC-VSC-OPF by imposing a lower bound on the VSI. Then the NN-based surrogate model is constructed to incorporate the voltage stability constraint into the optimization process explicitly.  
\vspace{5pt}
\subsection{Formulation of CC-VSC-OPF}
Consider an electrical network constituted by a set \(\mathcal{N}\) of \(N=\left|\mathcal{N}\right|\) buses and a set \(\mathcal{L}\) of \(L=\left|\mathcal{L}\right|\) lines. Subscripts $_{PQ}$, $_{PV}$, and $_{ref}$ are added to classify three types of buses, namely PQ, PV, and reference buses.\textcolor{black}{\footnote{\textcolor{black}{Detailed definitions of three types of buses should be referred to \cite{glover2012power}.}}} Let $v_i$ and $\theta_i$ denote the voltage magnitude and the voltage angle at each bus. To avoid messy formulas, it is assumed that the power injection at each bus $i\in\mathcal{N}$ can be decomposed into one controllable generation $p_{g,i},q_{g,i}$, one fixed demand $p_{d,i},q_{d,i}$, and one uncertain renewable generation $\xi_{i}$\footnote{The renewable energy stations are assumed to be operated at unity power factor here. However, other types of control can be incorporated as well.}. \textcolor{black}{Let a random vector $\boldsymbol{\xi}\coloneqq[\xi_1,\xi_2,\dots,\xi_N]^{\mathsf{T}}$ comprise uncertain renewable generations at each bus. Considering the interdependence between system variables (i.e., $p_{g,i},q_{g,i},v_i$ and $\theta_i$) and $\boldsymbol{\xi}$, during real-time operations, the values of these variables become dependent on the realizations of $\boldsymbol{\xi}$.\textcolor{black}{\footnote{\textcolor{black}{For example, according to the Automatic Generation Control, the power mismatch will cause changes in the power output of generators at PV and reference buses. More details can be found in \cite{roald2018chance}.}}} This implies that all these variables can be expressed as functions of $\boldsymbol{\xi}$.} Thus, the net active power injection at bus \textit{i} can be expressed as follows:
\begin{equation}\label{eq:1}
    p_i(\boldsymbol{\xi})=p_{g,i}(\boldsymbol{\xi})+\xi_{i}-p_{d,i},\quad \forall i \in \mathcal{N},
\end{equation}
Similarly, the net reactive power injection at bus \textit{i} is given by:
\begin{equation}\label{eq:2}
    q_i(\boldsymbol{\xi})=q_{g,i}(\boldsymbol{\xi})-q_{d,i},\quad \forall i \in \mathcal{N},
\end{equation}

The electrical network is governed by the AC nodal power balance equations, which are expressed as functions of bus injections and bus voltages in polar coordinates: 
\begin{subequations}\label{eq:4}
    \begin{alignat}{2}        
    p_i(\boldsymbol{\xi})=&v_i(\boldsymbol{\xi})\sum\limits_{j=1}^{N}v_j(\boldsymbol{\xi})\Bigl(G_{ij}\cos(\theta_{i}(\boldsymbol{\xi})-\theta_{j}(\boldsymbol{\xi})) \notag\\
    &+B_{ij}\sin(\theta_{i}(\boldsymbol{\xi})-\theta_{j}(\boldsymbol{\xi}))\Bigr),\quad \forall i \in\mathcal{N}, \\
    q_i(\boldsymbol{\xi})=&v_i(\boldsymbol{\xi})\sum\limits_{j=1}^{N}v_j(\boldsymbol{\xi})\Bigl(G_{ij}\sin(\theta_{i}(\boldsymbol{\xi})-\theta_{j}(\boldsymbol{\xi})) \notag\\
    &-B_{ij}\cos(\theta_{i}(\boldsymbol{\xi})-\theta_{j}(\boldsymbol{\xi}))\Bigr),\quad \forall i \in\mathcal{N},
    \end{alignat}
\end{subequations}
where $G_{ij}$ and $B_{ij}$ represent the real and imaginary components of the bus admittance matrix.

At this point, we have all the ingredients to state the CC-VSC-OPF problem.
\begin{subequations}\label{eq:5}
\allowdisplaybreaks
    \begin{alignat}{3}
        &\!\min &\quad&\sum\limits_{i\in\mathcal{N}_{PV,ref}}\mathbb{E}\left[C_i(p_{g,i})\right]&\quad &  \label{eq:5a}\\
        &\rm{s.t.}  &\quad &\boldsymbol{g}(\boldsymbol{p}(\boldsymbol{\xi}),\boldsymbol{q}(\boldsymbol{\xi}),\boldsymbol{v}(\boldsymbol{\xi}),\boldsymbol{\theta}(\boldsymbol{\xi}))=0,&\quad& \label{eq:5b}\\
        & &\quad&v_i^{\min}\leq\ v_i\leq v_i^{\max}, &\quad&\forall i \in \mathcal{N}_{PV}, \label{eq:5c}\\
        & &\quad&\mathbb{P}(p_{g,i}(\boldsymbol{\xi})\geq p_{g,i}^{\min})\geq 1-\epsilon_P,&\quad &\forall i \in \mathcal{N}_{PV,ref}, \label{eq:5d}\\
        & &\quad&\mathbb{P}(p_{g,i}(\boldsymbol{\xi})\leq p_{g,i}^{\max})\geq 1-\epsilon_P,&\quad &\forall i \in \mathcal{N}_{PV,ref}, \label{eq:5e}\\
        & &\quad &\mathbb{P}(q_{g,i}(\boldsymbol{\xi})\geq q_{g,i}^{\min})\geq 1-\epsilon_Q,&\quad &\forall i \in \mathcal{N}_{PV,ref},\label{eq:5f}\\
        & &\quad &\mathbb{P}(q_{g,i}(\boldsymbol{\xi})\leq q_{g,i}^{\max})\geq 1-\epsilon_Q,&\quad &\forall i \in \mathcal{N}_{PV,ref},\label{eq:5g}\\
        &  &\quad &\mathbb{P}(v_{i}(\boldsymbol{\xi})\geq v_{i}^{\min})\geq 1-\epsilon_V,&\quad &\forall i \in \mathcal{N}_{PQ},\label{eq:5h}\\
        &  &\quad &\mathbb{P}(v_{i}(\boldsymbol{\xi})\leq v_{i}^{\max})\geq 1-\epsilon_V,&\quad &\forall i \in \mathcal{N}_{PQ},\label{eq:5i}\\
        &  &\quad &\mathbb{P}(i_{ij}(\boldsymbol{\xi})\leq i_{ij}^{\max})\geq 1-\epsilon_I, &\quad &\forall ij\in\mathcal{L}, \label{eq:5j}\\
        &  &\quad &\mathbb{P}(\sigma(\boldsymbol{\xi})\geq \sigma^{\min})\geq 1-\epsilon_{\sigma}, &\quad & \label{eq:5k}
    \end{alignat}
\end{subequations}
\textcolor{black}{where $\mathcal{N}_{PV,ref}$ represents the bus set including all PV and reference buses, $i_{ij}$ denotes the current magnitude from bus $i$ to bus $j$, and $\sigma$ denotes the VSI.}

The objective \eqref{eq:5a} minimizes the total expected cost of active power generation. \eqref{eq:5b} represents the compact form of nodal power balance equations. \eqref{eq:5c} places constraints on the voltage setpoints of PV buses. \eqref{eq:5d}-\eqref{eq:5k} impose chance constraints on other variables subject to uncertainties 
with a desired level of confidence (or reliability)
$1-\epsilon_{P/Q/V/I/\sigma}$ for some small quantities $\epsilon_{P/Q/V/I/\sigma}$. \textcolor{black}{Specifically, \eqref{eq:5d}-\eqref{eq:5g} constitute physical limits for the active and reactive power output of each generator. Similar to \eqref{eq:5c}, \eqref{eq:5h}-\eqref{eq:5i} enforces voltage magnitude limits on PQ buses. In order to prevent damage to transmission elements resulting from excessive currents, a maximum limit is imposed on the current magnitude value for each branch in \eqref{eq:5j}. \eqref{eq:5k} guarantees a minimum distance from the voltage instability point, where the VSI is denoted by $\sigma(\boldsymbol{\xi})$. It is acknowledged that voltage stability is primarily concerned with the singularity of the power flow Jacobian $\boldsymbol{J}$ \textcolor{black}{\cite{bompard1996dynamic}}. In this work, the VSI $\sigma$ is defined as the minimum singular value (MSV) of $\boldsymbol{J}$ since this value tends to zero as the system approaches the voltage collapse point \textcolor{black}{\cite{lof1992fast}}.}
\begin{remark}
    \textcolor{black}{Despite intuitive, the main drawback of the above MSV-based representation of VSI is that the MSV is an implicit function of system variables}, which renders the constraint \eqref{eq:5k} challenging to be directly embedded into the OPF problem.
\end{remark}
\subsection{Surrogate Model for Voltage Stability Index}
An alternative feasible approach is to establish a high-precision surrogate model to accurately express the VSI. Owing to its strong fitting ability, the Back Propagation NN (BPNN) is applied to obtain the approximate analytic expression of $\sigma$ in this study. 
Before proceeding with the establishment of the NN-based surrogate model, we first introduce two variable sets: The first set $\boldsymbol{y}\coloneqq\{p_{g,i},v_i\mid i\in\mathcal{N}_{PV,ref}\}$ contains all control variables and the other set $\boldsymbol{z}\coloneqq\{v_i,\theta_j\mid i\in\mathcal{N},j\in\mathcal{N}_{PV,PQ}\}$ includes all voltage magnitudes and angles.  

A high-quality NN model hinges on a large quantity of properly sampled training data. To this end, the Latin hypercube sampling (LHS) method is employed to generate sufficient eligible samples of $\boldsymbol{y}$. Subsequently, standard power flow analysis is carried out to determine the corresponding $\boldsymbol{z}$ and $\sigma$ samples.
Since the Jacobian $\boldsymbol{J}$ is directly related to voltage magnitudes and angles, this study adopts the samples of $\boldsymbol{z}$ and $\sigma$ as the voltage stability dataset $D=\{(\boldsymbol{z},\sigma)\}$. \textcolor{black}{Here, we provide the summarized steps as follows and the corresponding flowchart in Fig. \ref{fig:9} for the readers' convenience.}

\begin{figure}[!t]
\centering
\includegraphics[width=3in]{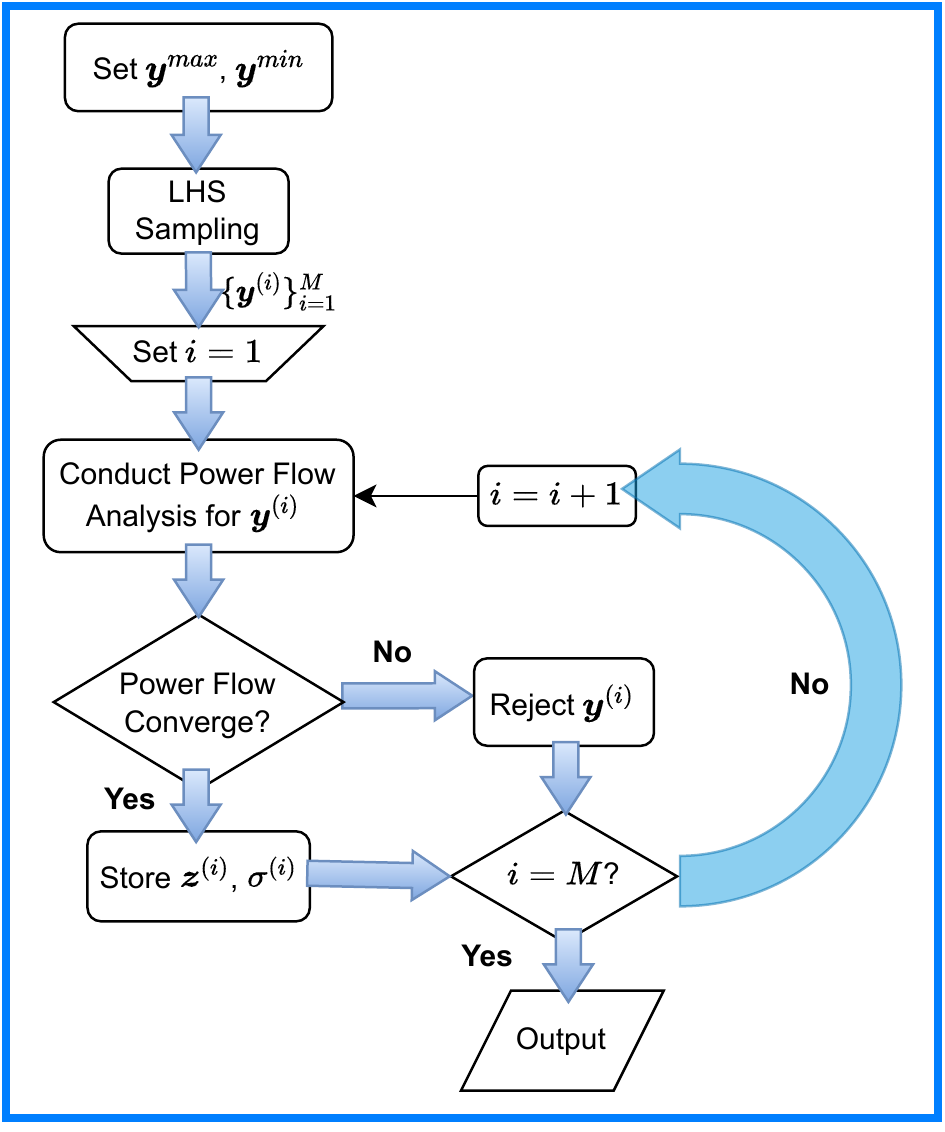}
\caption{\textcolor{black}{Flowchart for establishing NN-based model}}
\label{fig:9}
\vspace{-15pt}
\end{figure}
\textcolor{black}{\textit{Step 1:} LHS sampling. Use LHS to generate abundant samples $\{\boldsymbol{y}^{(i)}\}_{i=1}^{M}$. Here, each variable $y_i$ is assumed to follow a uniform distribution between its minimum and maximum allowable values, $y_i^{\min}$ and $y_i^{\max}$.} 

\textcolor{black}{\textit{Step 2:} Sample rejection. Note that not every sample in $\{\boldsymbol{y}^{(i)}\}_{i=1}^{M}$ is guaranteed to belong to the feasibility set of AC power flow equations \eqref{eq:4}. Hence, it is imperative to execute the power flow feasibility test for each $\boldsymbol{y}^{(i)}$ and discard any samples falling outside the feasible space.\footnote{\textcolor{black}{This necessitates solving the AC power flow problem by assigning $\boldsymbol{y}^{(i)}$ as the setpoint for control variables. If the power flow calculation fails to converge, then the associated sample $\boldsymbol{y}^{(i)}$ is rejected; otherwise, the corresponding values for $\boldsymbol{z}$ and $\sigma$ are retained.}} This process continues until all samples have been checked, ultimately yielding the collection of feasible samples $\{\boldsymbol{y}^{(i)},\boldsymbol{z}^{(i)},\sigma^{(i)}\}_{i=1}^{N_s}$.}

\textcolor{black}{\textit{Step 3:} NN training. Utilize $\{\boldsymbol{z}^{(i)},\sigma^{(i)}\}_{i=1}^{N_s}$ to train a BPNN model with $\boldsymbol{z}$ being predictors and $\sigma$ being the response. After training, the explicit expression for $\sigma$ is given by a combination of function composition and matrix multiplication:
\begin{equation}\label{eq:6}
        \sigma(\boldsymbol{z})\coloneqq f^{L}(\boldsymbol{W}^{L}\boldsymbol{f}^{L-1}(\boldsymbol{W}^{L-1}\cdots\boldsymbol{f}^{1}(\boldsymbol{W}^{1}\boldsymbol{z})\cdots)),
\end{equation}
where $L$ is the number of layers, $\boldsymbol{W}^{L}=(w_{jk}^l)$ are weights between layer $l-1$ and $l$, and $\boldsymbol{f}^l$ is the activation function at layer $l$.}

\textcolor{black}{At this stage, the deterministic VSC-OPF can be directly handled by common nonlinear optimization methods, such as stochastic gradient descent (SGD) and interior point method (IPM), following the substitution of $\sigma$ with \eqref{eq:6}; however, the inclusion of uncertain renewable generations introduces a challenging aspect of managing chance constraints. Moreover, the selection between SGD and IPM for solving VSC-OPF is important. Despite the computational advantages and lower memory requirements attributed to SGD, its convergence performance may prove unsatisfactory \cite{10.5555/2969442.2969497}. In contrast, IPM exhibits superior robustness and requires much fewer iterations for convergence, making it widely employed as the benchmark in the existing literature \cite{venzke2020inexact,liu2019acopf,mhanna2018adaptive}. Therefore, this paper opts to employ IPM for problem resolution. It should be noted that} the IPM requires the hessian of \eqref{eq:6}, which will be prohibitively time-consuming and memory-intensive for large-scale systems. These issues will be addressed in subsequent sections.
\section{Practical Uncertainty Propagation Approach}\label{approach}
The main focus of this section is to provide a novel approach known as APCE that effectively deals with chance constraints and transcends the limitations of existing PCE methods.   
\subsection{Motivation}
\textcolor{black}{Let the variable set $\boldsymbol{x}\coloneqq\{p_{g,i},q_{g, i},v_j,i_{ij},\sigma\}$ comprise all uncertainty-affected variables in \eqref{eq:5d}-\eqref{eq:5k}.} The implicit functional dependency between $\boldsymbol{x}$ and the uncertain renewable generations $\boldsymbol{\xi}$ \textcolor{black}{at a fixed operation point $\boldsymbol{y}^{fixed}$} can be defined in the form of:
\begin{equation}
    \textcolor{black}{\boldsymbol{x}=\boldsymbol{h}(\boldsymbol{y}^{fixed},\boldsymbol{\xi})}.
\end{equation}

It is mathematically complex or even impossible to analytically derive the probability distribution of $\boldsymbol{x}$. In practice, the truncated PCE technique can be adopted as a surrogate model to approximate $x_i$ as a finite sum of polynomials of $\boldsymbol{\xi}$ \cite{feinberg2018multivariate}:
\begin{equation}
    \hat{x}_i=\sum\limits_{j=0}^Ma_j\boldsymbol{\Psi}_j(\boldsymbol{\xi}),
\end{equation}
where $\{\boldsymbol{\Psi}_j(\boldsymbol{\xi})\}$ is the set of multivariate orthogonal polynomial basis and $\{a_j\}$ is the set of unknown coefficients. The crux of PCE is the identification of the basis for dependent random variables. Although Nataf transformation has been adopted to transform dependent variables into independent ones \cite{xu2021iterative}, it assumes a Gaussian copula as the underlying dependence structure, which might not align with the real case. 
\begin{figure}[!t]
\subfloat[]{\includegraphics[width=1.65in]{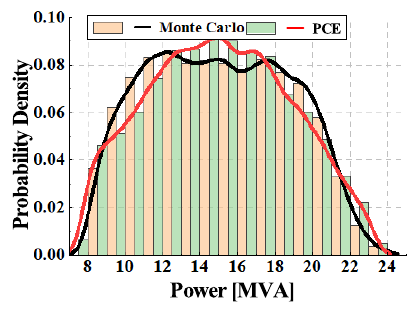}}
\quad
\subfloat[]{\includegraphics[width=1.65in]{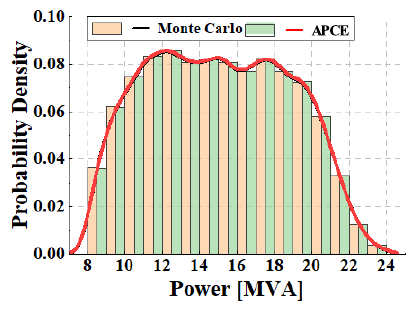}}
\caption{Comparison between MC simulations and: (a) PCE; (b) APCE}
\label{fig:1}
\vspace{-15pt}
\end{figure}
This claim is exemplified in Fig. \ref{fig:1}, which compares the line flow distributions calculated with the MC simulations, the Gaussian-copula-based PCE, and the APCE that will be covered in the succeeding section. It is obvious that the proposed APCE outperforms the other as it accurately captures the propagated line flow distribution. 
\subsection{Adaptive Polynomial Chaos Expansion}\label{APCE}
This part will present how to construct the orthonormal polynomial basis for dependent variables through the APCE method \textcolor{black}{\cite{lee2020practical}}.
It is important not to conflate the APCE with the generalized PCE in \cite{xiu2002wiener} and the arbitrary PCE in \cite{oladyshkin2012data} since the latter two approaches are both incapable of handling nonlinear dependence.
The salient feature of APCE lies in its capability to inherently accommodate the interdependence between different renewable generations and arbitrary marginal probability distributions under very mild assumptions. As a result, the APCE is highly flexible and adaptive, boasting decent accuracy in uncertainty propagation, as verified in Fig. \ref{fig:1}. Moreover, the proposed APCE can be implemented as entirely data-driven, eliminating the need for estimating either the marginal distribution or the joint probability distribution.
\subsubsection{Gram Matrix}
To start with, we first develop a set of monomial basis that can be subsequently combined into orthogonal polynomials.
Recall that given an $N$-dimensional multi-index $\boldsymbol{j}\coloneqq[j_1,\dots,j_N]\in\mathbb{N}^N$, a monomial of the random vector $\boldsymbol{\xi}$ is defined as the following product type:
\begin{equation}
    \boldsymbol{\xi}^{\boldsymbol{j}}=\xi_1^{j_1}\dots\xi_N^{j_N},
\end{equation}
which has a total degree $\vert\boldsymbol{j}\vert=j_1+\dots+j_N$.

Considering a fixed $u\in\mathbb{N}$, let $\mathcal{P}_u^N$ denote the multi-index set consisting of all $N$-dimensional multi-indices $\boldsymbol{j}\in\mathbb{N}^N$ such that $\vert\boldsymbol{j}\vert=u$. In other words, $\mathcal{P}_u^N$ is the set of all multi-indices with a fixed norm of $u$.
Then, it is elementary to show that the cardinality of $\mathcal{P}_u^N$ is $\binom{N+u-1}{u}$.
Subsequently, the cardinality of the multi-index set with the norm being at most $m$ can be calculated by summing up the cardinalities of the sets $\mathcal{P}_u^N$ for all $u\in\mathbb{N}$: 
\begin{equation}\label{eq:10}
    L_{m}=\sum\limits_{u=0}^m\binom{N+u-1}{u}=\binom{N+m}{m}.
\end{equation}
Denote by $\boldsymbol{j}^{(1)},\dots,\boldsymbol{j}^{(L_{m})}$ the sorted elements of this multi-index set using the graded reverse lexicographical ordering, then the corresponding monomials of $\boldsymbol{\xi}$ are defined by $\boldsymbol{\xi}^{\boldsymbol{j}^{(1)}},\dots,\boldsymbol{\xi}^{\boldsymbol{j}^{(L_{m})}}$, which also constitute the basis of the space of homogeneous polynomials with a degree at most $m$.

Gathering the above basis yields an \textcolor{black}{$L_{m}$}-dimensional monomial vector with the $i$-th element being $\boldsymbol{\xi}^{\boldsymbol{j}^{(i)}}$, i.e., $\boldsymbol{P}_m(\boldsymbol{\xi})=[\boldsymbol{\xi}^{\boldsymbol{j}^{(1)}},\dots,\boldsymbol{\xi}^{\boldsymbol{j}^{(L_{m})}}]^{\mathsf{T}}$. Thereafter, we further define the Gram matrix for monomials as follows.
\begin{definition}[Gram matrix]\label{def:1}
Given a random vector $\boldsymbol{\xi}$ and the corresponding monomial vector $\boldsymbol{P}_m(\boldsymbol{\xi})$, the Gram matrix is a matrix $\boldsymbol{G}_{m}$ whose $(p,q)$-th entry is the inner product of the $p$-th and $q$-th element of $\boldsymbol{P}_m(\boldsymbol{\xi})$:
\begin{equation}
    G_{m,pq}=\left\langle \boldsymbol{\xi}^{\boldsymbol{j}^{(p)}}, \boldsymbol{\xi}^{\boldsymbol{j}^{(q)}} \right\rangle=\int \boldsymbol{\xi}^{\boldsymbol{j}^{(p)}+\boldsymbol{j}^{(q)}}f_{\boldsymbol{\xi}}(\boldsymbol{\xi}) \, d\boldsymbol{\xi},
\end{equation}
where $f_{\boldsymbol{\xi}}(\boldsymbol{\xi})$ denotes the joint probability density function.
\end{definition}
In general, $f_{\boldsymbol{\xi}}(\boldsymbol{\xi})$ proves to be unavailable or challenging to estimate. To address this challenge, a data-driven method can be leveraged as an alternative:
\begin{equation}\label{eq:12}
    G_{m,pq}\approx\frac{1}{N_s}\sum\limits_{l=1}^{N_s}(\xi_1^{(l)})^{j_1^{(p)}+j_1^{(q)}}\dots(\xi_N^{(l)})^{j_N^{(p)}+j_N^{(q)}},
\end{equation}
where $\boldsymbol{\xi}^{(l)}=[\xi_1^{(l)},\dots,\xi_N^{(l)}]^{\mathsf{T}}$ denotes the $l$-th sample.
\subsubsection{Orthogonal Basis}
It is evident that $\boldsymbol{G}_m$ is positive-definite, thereby permitting the whitening transformation to generate orthogonal polynomials. 
In this paper, the Cholesky whitening matrix $\boldsymbol{W}_m$ is adopted which satisfies the condition $\boldsymbol{W}_m^{\mathsf{T}}\boldsymbol{W}_m=\boldsymbol{G}_m^{-1}$. This approach entails computing the Cholesky decomposition of $\boldsymbol{G}_m$ to obtain a lower-triangular matrix $\boldsymbol{U}_m$, which satisfies $\boldsymbol{G}_m=\boldsymbol{U}_m\boldsymbol{U}_m^\mathsf{T}$. The inverse of $\boldsymbol{U}_m$ then serves as a whitening matrix $\boldsymbol{W}_m$, which yields an $L_{m}$-dimensional polynomial vector:
\begin{equation}\label{eq:13}
    \boldsymbol{\Psi}_m(\boldsymbol{\xi})=\boldsymbol{W}_m\boldsymbol{P}_m(\boldsymbol{\xi})=\boldsymbol{U}_m^{-1}\boldsymbol{P}_m(\boldsymbol{\xi}).
\end{equation}
Next, we would like to state the appealing property of this new set of polynomials, as presented in Proposition \ref{prop:1}.
\begin{prop}[Orthonormality]\label{prop:1}
For any distinct $i,j\in\mathbb{N}\cap[1,L_{m}]$, the $i$-th and $j$-th element of the polynomial vector $\boldsymbol{\Psi}_m(\boldsymbol{\xi})$ are mutually orthonormal w.r.t. the probability measure of $\boldsymbol{\xi}$.
\end{prop}
\begin{proof}
With a slight abuse of notation, let $\Psi_i(\boldsymbol{\xi})$ denote the $i$-th component of the vector $\boldsymbol{\Psi}_m(\boldsymbol{\xi})$. To prove the orthonormality between $\Psi_i(\boldsymbol{\xi})$ and $\Psi_j(\boldsymbol{\xi})$, it suffices to show $\left\langle \Psi_i(\boldsymbol{\xi}), \Psi_j(\boldsymbol{\xi}) \right\rangle_{\boldsymbol{\xi}}=\delta_{ij}$, where $\delta_{ij}$ is the Kronecker delta. 
From Definition \ref{def:1}, it can be seen that $\boldsymbol{G}_m=\mathbb{E}\left[\boldsymbol{P}_m(\boldsymbol{\xi})\boldsymbol{P}_m^{\mathsf{T}}(\boldsymbol{\xi})\right]$. Hence, the second-order moment of $\boldsymbol{\Psi}_m(\boldsymbol{\xi})$ can be expressed and simplified by
\begin{equation}\label{eq:14}
\begin{split}
    \mathbb{E}\left[\boldsymbol{\Psi}_m(\boldsymbol{\xi})\boldsymbol{\Psi}_m^{\mathsf{T}}(\boldsymbol{\xi})\right]
    & = \mathbb{E}\left[\boldsymbol{W}_m\boldsymbol{P}_m(\boldsymbol{\xi})\boldsymbol{P}_m^{\mathsf{T}}(\boldsymbol{\xi})\boldsymbol{W}_m^{\mathsf{T}}\right] \\
    & = \boldsymbol{W}_m\mathbb{E}\left[\boldsymbol{P}_m(\boldsymbol{\xi})\boldsymbol{P}_m^{\mathsf{T}}(\boldsymbol{\xi})\right]\boldsymbol{W}_m^{\mathsf{T}} \\
    & = \boldsymbol{W}_m\boldsymbol{G}_m\boldsymbol{W}_m^{\mathsf{T}} = \boldsymbol{I}_m,
    \end{split}
\end{equation}
where $\boldsymbol{I}_m$ is the $L_{m}$-order identity matrix. Therefore, the inner product between $\Psi_i(\boldsymbol{\xi})$ and $\Psi_j(\boldsymbol{\xi})$ is
\begin{equation}
\begin{split}
        \left\langle \Psi_i(\boldsymbol{\xi}), \Psi_j(\boldsymbol{\xi}) \right\rangle_{\boldsymbol{\xi}}& = \int \Psi_i(\boldsymbol{\xi})\Psi_j(\boldsymbol{\xi})f_{\boldsymbol{\xi}}(\boldsymbol{\xi}) \, d\boldsymbol{\xi}\\
        & = \mathbb{E}\left[\Psi_i(\boldsymbol{\xi})\Psi_j(\boldsymbol{\xi})\right]=\delta_{ij},
\end{split}
\end{equation}
which completes the proof.
\end{proof}\vspace{-8pt}
\begin{remark}
The whitening transformation is fundamentally distinct from the isoprobabilistic transformation, such as the Rosenblatt transformation. The latter approach aims to convert dependent random variables into independent ones, relying on full knowledge of the joint probability distribution. In contrast, the whitening transformation represents a fresh perspective that utilizes a linear combination of the monomial basis to generate an orthonormal polynomial basis w.r.t. the original dependent probability measure $f_{\boldsymbol{\xi}}(\boldsymbol{\xi})$.
\end{remark}
\subsubsection{APCE Approximation}
Note that the dimension of $\boldsymbol{\Psi}_m$ is consistent with the polynomial space of degree at most $m$, implying that this set of orthogonal polynomials functions as a basis. As $m$ approaches infinity, the vector $\boldsymbol{\Psi}_m$ encompasses an infinite collection of orthonormal bases and thus enables the Fourier polynomial expansion of the random vector $\boldsymbol{x}$, which is the APCE, as rigorously stated in Theorem \ref{the:1}.
\begin{theorem}\label{the:1}
%Let $\left(\Omega,\mathcal{F},\mathbb{P}\right)$ be the probability space and the random vector $\boldsymbol{\xi}$ be the mapping from $\Omega$ to $\mathbb{R}^N$. 
Assume that the probability density function of $\boldsymbol{\xi}$ is compactly supported or is exponentially integrable. Given a set of multivariate orthonormal polynomials $\{\Psi_i(\boldsymbol{\xi}),1\leq i<\infty\}$ w.r.t. $f_{\boldsymbol{\xi}}(\boldsymbol{\xi})d\boldsymbol{\xi}$, then each finite variance computation model represented as a map \textcolor{black}{$x_i=h_i(\boldsymbol{y}^{fixed},\boldsymbol{\xi})$} can be expanded as a Fourier series composed of $\{\Psi_i(\boldsymbol{\xi})\}$:  
\begin{equation}\label{eq:16}
    \textcolor{black}{x_i=h_i(\boldsymbol{y}^{fixed},\boldsymbol{\xi})\sim \sum\limits_{j=1}^{\infty}a_{i,j}\Psi_j(\boldsymbol{\xi})},
\end{equation}
%where $a_{i,j},i=1,\dots,N,j=1,\dots,\infty$ are the expansion coefficients and 
where the symbol $\sim$ means the APCE of $x_i$ converges in mean-square, both in probability and distribution.
The coefficients $a_{i,j},i=1,\dots,N,j=1,\dots,\infty$ are defined as
\begin{equation}
    \textcolor{black}{a_{i,j}=\int h_i(\boldsymbol{y}^{fixed},\boldsymbol{\xi})\Psi_j(\boldsymbol{\xi})f_{\boldsymbol{\xi}}(\boldsymbol{\xi}) \, d\boldsymbol{\xi}}.
\end{equation}
\end{theorem}
\begin{proof}
    Detailed proof can be found in \cite{rahman2018polynomial}.
\end{proof}\vspace{-5pt}
\textcolor{black}{In power systems, the function $h_i(\boldsymbol{y}^{fixed},\boldsymbol{\xi})$ is a non-polynomial mapping, which means it can only be exactly represented by a linear combination of an infinite polynomial basis, which is impractical for engineering applications. Therefore}, \eqref{eq:16} must be truncated into a finite sum via the total-degree truncation scheme. In this work, an $m$-th order APCE approximation comprises all polynomials of total degree less than or equal to $m$:
\begin{equation}\label{eq:18}
    \hat{x}_i=\sum\limits_{j=1}^{L_{m}}a_{i,j}\Psi_j(\boldsymbol{\xi}).
\end{equation}
The expansion coefficients can be calculated using a non-intrusive method based on ordinary least-squares regression, which is similar to other PCE methods and hence omitted here. Based on this established expansion, the empirical distribution of $x_i$ can thus be accessed inexpensively by feeding the renewable generation dataset $\{\boldsymbol{\xi}^{(i)}\}_{i=1}^{N_s}$ into \eqref{eq:18}.

\textcolor{black}{Currently, there is no golden rule for the selection of $m$ in \eqref{eq:18}. However, previous studies \cite{xu2021iterative,muhlpfordt2019chance} have indicated that a second-order traditional PCE model will be sufficiently accurate for characterizing power system variables. The setting of $m=2$ in the proposed APCE method will also be corroborated by test results in Section V.}
\subsection{Dimensionally Decomposed APCE}\label{DD-APCE}
It should be acknowledged that the APCE, like most existing PCE methods, will succumb to the curse of dimensionality, thereby limiting its application to high-dimensional cases. Motivated by the sparsity-of-effects principle, which suggests that a lower-variate approximation suffices to capture the uncertainty of the outputs, we further adopt a novel DD-APCE method \cite{lee2023high} to restructure the basis functions of the APCE in a dimensionwise manner.

The heart of the DD-APCE lies in limiting the interaction order among input variables. Following this intuition, for each $m\in \mathbb{N}$ and $0\leq s\leq N$, we first introduce the reduced multi-index set, denoted by
\begin{equation}\label{eq:19}
    \begin{split}
        \Pi_{s,m}^{N}\coloneqq\{&\boldsymbol{j}=(\boldsymbol{j}_l, \mathbf{0}_{-l})\in\mathbb{N}^N:\boldsymbol{j}_l\in\mathbb{N}^{\vert l\vert}, \\
        &\vert l\vert\leq\vert \boldsymbol{j}_l\vert\leq m, 0\leq\vert l\vert\leq s\},
    \end{split}
\end{equation}
where $l\subseteq\{1,\dots,N\}$ is a subset and $-l\coloneqq\{1,\dots,N\}\backslash l$ is the corresponding complementary subset. Here, $\boldsymbol{j}=(\boldsymbol{j}_l,\mathbf{0}_{-l})$ denotes an $N$-dimensional multi-index whose $i$-th component is $j_i$ if $i\in l$ and 0 otherwise. Hence, $\boldsymbol{j}$ has $\vert l\vert$ nonzero elements. Since $l$ has at most $s$ elements, $\Pi_{s,m}^N$ contains all multi-indices with at most $s$ nonzero elements and with the norm being at most $m$. The cardinality of $\Pi_{s,m}^N$ is given by
\begin{equation}\label{eq:20}
    L_{s,m}\coloneqq 1+\sum\limits_{k=1}^s\binom{N}{k}\binom{m}{k}.
\end{equation}

With the graded reverse lexicographical ordering, the elements of $\Pi_{s,m}^N$ are arranged as $\boldsymbol{j}^{(1)},\dots,\boldsymbol{j}^{(L_{s,m})}$ and the corresponding monomials of $\boldsymbol{\xi}$ are denoted by $\boldsymbol{\xi}^{\boldsymbol{j}^{(1)}},\dots,\boldsymbol{\xi}^{\boldsymbol{j}^{(L_{s,m})}}$. Employing the same procedures as described in Section \ref{APCE}, we can calculate the Gram matrix $\boldsymbol{G}_{s,m}$ and apply a whitening transformation to this new set of monomials to generate a set of measure-consistent orthonormal polynomials $\boldsymbol{\Psi}_{s,m}(\boldsymbol{\xi})$. This, in turn, leads to an $s$-variate, $m$-th order DD-APCE approximation of the form:
\begin{equation}
    \hat{x}_i=\sum\limits_{j=1}^{L_{s,m}}a_{i,j}\Psi_j(\boldsymbol{\xi}).
\end{equation}

Compared to the regular APCE, the DD-APCE, with $L_{s,m}\leq L_{m}$, offers superior efficiency while maintaining acceptable accuracy. In particular, this advantage is expected to be more pronounced in high-dimensional cases. To illustrate, the number of basis functions in a second-order APCE approximation scales quadratically as $(N+2)(N+1)/2$, whereas that of the uni-variate DD-APCE grows only linearly as $2N+1$. 
Therefore, the DD-APCE is capable of generating measure-consistent orthonormal polynomials in a stable manner while exhibiting exceptional efficacy in high dimensions. These attributes render the DD-APCE a promising candidate for practical uncertainty quantification in power systems.
\section{The Integrated Data-driven Scheme}\label{scheme}
In the previous section, we discussed the characterization of model responses under complexly correlated uncertainties. Our current goal is to amalgamate all the aforementioned techniques into a data-driven scheme tailored to tackle CC-VSC-OPF problems. However, prior to accomplishing this objective, it is imperative to address the computational burden associated with \eqref{eq:6}, thereby enhancing the feasibility and practicality of the proposed approach.
\subsection{PLS-NN Framework}
Recall that the IPM necessitates the Hessian of \eqref{eq:6} when solving an ordinary VSC-OPF problem. Nevertheless, when dealing with a large-scale NN, deriving the analytical expression for the Hessian matrix can be a laborious and complex task. Furthermore, even the numerical evaluation of the Hessian matrix becomes impractical when the dimension $2N-1$ exceeds a certain threshold.

One workaround to overcome this deficiency is to compute the gradient and Hessian of a reduced set of latent variables $\boldsymbol{\upsilon}$ and then retrieve the derivatives w.r.t. the original predictor variables $\boldsymbol{z}$. To facilitate this inverse transform, linear dimension reduction techniques are preferred, which render the following generic formula:
\begin{equation}\label{eq:21}
    \boldsymbol{\upsilon}=\boldsymbol{g}(\boldsymbol{z})=\boldsymbol{z}\boldsymbol{P},
\end{equation}
where the dimensions of $\boldsymbol{\upsilon}, \boldsymbol{z}, \boldsymbol{P}$ are $1\times M, 1\times (2N-1)$ and $(2N-1)\times M$, respectively. The NN surrogate model is then trained with $\boldsymbol{\upsilon}$ as input instead of $\boldsymbol{z}$, and $\sigma$ as output, yielding a more parsimonious model $\hat{\sigma}(\boldsymbol{\upsilon})$. During each iteration of the IPM procedure, the evaluation of $\sigma(\boldsymbol{z})$, its gradient, and Hessian for a given value of $\boldsymbol{z}^{(i)}$ can be derived as:
\begin{subequations}\label{eq:22}
    \begin{alignat}{1}
            \sigma(\boldsymbol{z}^{(i)}) &=\hat{\sigma}(\boldsymbol{\upsilon}^{(i)})=\hat{\sigma}(\boldsymbol{z}^{(i)}\boldsymbol{P}),\\
        \nabla \sigma\mid_{\boldsymbol{z}^{(i)}}&=\nabla (\hat{\sigma}\circ \boldsymbol{g})\mid_{\boldsymbol{z}^{(i)}} \notag\\
        &=\boldsymbol{P}\nabla \hat{\sigma}\mid_{\boldsymbol{\upsilon}^{(i)}},\\
        \nabla^2\sigma\mid_{\boldsymbol{z}^{(i)}}&=\nabla^2 (\hat{\sigma}\circ \boldsymbol{g})\mid_{\boldsymbol{z}^{(i)}} \notag\\
        &=\boldsymbol{P}\nabla^2 \hat{\sigma}\mid_{\boldsymbol{\upsilon}^{(i)}} \boldsymbol{P}^{\mathsf{T}}.
    \end{alignat}
    \vspace{-6pt}
\end{subequations}
\begin{figure}[!t]
\removelatexerror
\begin{algorithm}[H]
    \caption{The PLS-NN Framework}
    Generate abundant samples $\{\boldsymbol{y}^{(i)}\}_{i=1}^{M}$ using LHS\;
    Obtain feasible samples $\{\boldsymbol{y}^{(i)},\boldsymbol{z}^{(i)},\sigma^{(i)}\}_{i=1}^{N_s}$ through power flow analysis \;
    \eIf{\rm the value of $2N-1$ is reasonable}
    {Train the NN-based surrogate model $\sigma(\boldsymbol{z})$ as shown in \eqref{eq:6}\;}
    {Choose the dimension of latent variables $\boldsymbol{\upsilon}$ \;
    Perform PLS on voltage stability dataset $D$ and obtain the transformation matrix $\boldsymbol{P}$\;
    Construct latent variables through \eqref{eq:21} and create dataset $\{(\boldsymbol{\upsilon},\sigma)\}$\;
    Train the NN-based surrogate model, which yields the reduced model $\hat{\sigma}(\boldsymbol{\upsilon})$.}
        \end{algorithm}
        \vspace{-10pt}
\end{figure}

As indicated by \eqref{eq:22}, the quality of the retrieved values, $\sigma(\boldsymbol{z})$, $\nabla \sigma(\boldsymbol{z})$ and $\nabla^2 \sigma(\boldsymbol{z})$, is contingent upon the prediction accuracy of the reduced model $\hat{\sigma}(\boldsymbol{\upsilon})$. To this end, PLS is proposed in this study to construct new predictor variables $\boldsymbol{\upsilon}$ \cite{de1993simpls}. As a supervised technique, PLS maximizes the predictive power of $\boldsymbol{\upsilon}$ by accounting for the variability in the response. By capturing the variables that are highly informative regarding the response, PLS serves as a more effective alternative to unsupervised methods like principal component analysis and is expected to attain a parsimonious surrogate model $\hat{\sigma}(\boldsymbol{\upsilon})$ with satisfactory accuracy. 
A detailed procedure for implementing the proposed PLS-NN framework is presented in Algorithm 1. Within this framework, the computational overheads related to VSC-OPF problems are significantly reduced since it only requires a $M \times M$ Hessian as opposed to the original $(2N-1)\times (2N-1)$ one, where the $M \times M$ Hessian can be calculated inexpensively via several differentiation techniques, such as central difference differentiation. Meanwhile, this workaround not only enhances the scalability of the proposed approach but also serves to mitigate the limitations of the NN model training under high-dimensional inputs.
\begin{figure}[!t]
    \removelatexerror
     \begin{algorithm}[H]
             \KwIn{$\{\boldsymbol{\xi}^{(i)}\}_{i=1}^{N_s}$ - $N$-dimensional renewable generation dataset}
        \KwOut{$\boldsymbol{y}$ - Final solution of CC-VSC-OPF}
        \caption{The Iterative Data-Driven Scheme}
        Fix renewable generations at their expected values and then obtain the surrogate model of the VSI, $\sigma(\boldsymbol{z})\backslash\hat{\sigma}(\boldsymbol{\upsilon})$, according to Algorithm 1\;
        Build the deterministic VSC-OPF model\;
        Specify parameters of DD-APCE: total degree $m$ and maximum interaction order $s$\;
        Create $L_{s,m}$-dimensional monomial basis $\boldsymbol{P}_{s,m}(\boldsymbol{\xi })$\;
        Construct approximated $L_{s,m}\times L_{s,m}$ Gram matrix $\boldsymbol{G}_{s,m}$ via \eqref{eq:12}\;
        Generate orthonormal polynomial basis $\boldsymbol{\Psi}_{s,m}(\boldsymbol{\upsilon})$ through whitening transformation\;
               Set $k=0$, and initialize $\boldsymbol{x}_{\min}^{(0)}$ and $\boldsymbol{x}_{\max}^{(0)}$\; 
        Solve the deterministic VSC-OPF and obtain the initial operation point, $\boldsymbol{y}^{(0)}$ \;
        \While{$k\leq k_{\max}$}
        {
                                          Calculate coefficients of the $s$-variate, $m$-th order DD-APCE metamodel\;
                Approximate the distributions of $\boldsymbol{x}$ through stochastic testing\;
                Calculate margins $\Delta\boldsymbol{x}_{\min/\max}^{(k)}$ via \eqref{eq:26}\;
                \eIf{$\Delta\boldsymbol{x}_{\min}^{(k)}=\mathbf{0}$ \& $\Delta\boldsymbol{x}_{\max}^{(k)}=\mathbf{0}$}
                {
                \textbf{break}\;
                }{
                Update upper and lower bounds via \eqref{eq:27}\;
                Update iteration count $k\gets k+1$\;
                Solve the deterministic VSC-OPF to update the operation point $\boldsymbol{y}^{(k)}$
                }
                           
        }
     \end{algorithm}
     \vspace{-10pt}
\end{figure}
\subsection{The Iterative Scheme}
So far, we have developed the PLS-NN framework to solve large-scale VSC-OPF problems and have devised the DD-APCE technique to evaluate chance constraints at a fixed operation point. Our next step involves employing an iterative scheme inspired by \cite{roald2018chance} to merge these techniques, which enables the efficient solution of CC-VSC-OPF problems. The whole procedure is presented in Algorithm 2.

Prior to commencing the iteration process, the implicit VSI is substituted with the NN-based surrogate model. Subsequently, the deterministic VSC-OPF problem is formulated by assigning the uncertain renewable generations their anticipated values. At each iteration $k$, the deterministic VSC-OPF problem is solved utilizing the IPM, where the upper and lower bounds of $\boldsymbol{x}$ are set as $\boldsymbol{x}_{\min}^{(k)}$ and $\boldsymbol{x}_{\max}^{(k)}$, respectively. Once the operation point $\boldsymbol{y}^{(k)}$ is determined, the expansion coefficients in \eqref{eq:18} are calculated, enabling the estimation of empirical distributions of $\boldsymbol{x}$. 
Then, the tightening margins $\Delta\boldsymbol{x}_{\min}^{(k)}$ and $\Delta\boldsymbol{x}_{\max}^{(k)}$, which are introduced to ensure the prescribed violation probabilities, are calculated in a heuristic way:  
\begin{subequations}\label{eq:26}
    \begin{align}
        \Delta x_{i,\min}^{(k)}&=\left[x_{i,\min}^{(0)}-Q_{\epsilon_i}^{(k)}\left(x_i\right)\right]^+,\\
        \Delta x_{i,\max}^{(k)}&=\left[Q_{1-\epsilon_i}^{(k)}\left(x_i\right)-x_{i,\max}^{(0)}\right]^+,
    \end{align}
\end{subequations}
where $Q_{\alpha}^{(k)}(x_i)$ is the $\alpha$ quantile of the estimated distribution of $x_i$ at iteration $k$.
At the end of each iteration, the upper and lower bounds are updated via
\begin{subequations}\label{eq:27}
    \begin{align}
        \boldsymbol{x}_{\min}^{(k+1)}&=\boldsymbol{x}_{\min}^{(k)}+\Delta\boldsymbol{x}_{\min}^{(k)},\\
        \boldsymbol{x}_{\max}^{(k+1)}&=\boldsymbol{x}_{\max}^{(k)}-\Delta\boldsymbol{x}_{\max}^{(k)}.
    \end{align}
\end{subequations}

Following the above transformation, the CC-VSC-OPF problem is reduced to the traditional AC OPF with an explicit voltage stability constraint, which is amenable to existing optimization solvers.
Moreover, it can be seen that this algorithm is entirely data-driven and distribution-free, rendering it highly suitable for practical applications.
\textcolor{black}{
\begin{remark}
    It should be noted that the proposed Algorithm 2 cannot ensure global optimality under the problem setting. The nonconvex nature of the CC-VSC-OPF problem renders the optimization task at least NP-hard and there is no guarantee to obtain a global minimizer. However, as Roald and Anderson have tested in \cite{roald2018chance}, the iterative scheme can provide solutions close to a local optimum. This work aims to devise a high-quality and efficient solution framework for the CC-VSC-OPF problem, and global optimality is beyond the main focus.
\end{remark}
\begin{remark}
    Regarding the convergence property of the proposed iterative scheme, to the best of the authors' knowledge, the only research that rigorously conducts theory analysis is \cite{brust2022con}, wherein the researchers make the first attempt to prove a sufficient condition for convergence. Besides, a similar iterative scheme has been successfully implemented in \cite{roald2018chance,xu2021iterative}, and it also demonstrates satisfactory convergence performance in all test cases in Section V. Therefore, this iterative scheme is expected to perform well in practical applications.
\end{remark}
}
\section{Simulations and Results}\label{simulation}
In this section, various case studies are conducted on the IEEE 14-bus, IEEE 30-bus, European 89-bus, and Illinois 200-bus systems whose data can be obtained from MATPOWER \cite{zimmer2011matpower}. Solar power plants are added to the systems to introduce uncertainties in the VSC-OPF problem. The hourly power output of solar stations is simulated through the web platform, Renewables.ninja \cite{pfenninger2016long}. 
All the simulations are implemented on a desktop with 3.80-GHz AMD Ryzen 7 5800X processors and 16 GB RAM. The program environment is MATLAB 2021b. The VSC-OPF problem is implemented with the \textit{callback} function in MATPOWER and is resolved with Ipopt.
\vspace{-10pt}
\subsection{Small test cases}
This case study is performed on the IEEE 14-bus and IEEE 30-bus systems, with the corresponding voltage stability thresholds being 0.56 and 0.24. Three 60 MW solar power plants and five 30 MW solar power plants are added to the above two systems, respectively. A total of 11680 \textcolor{black}{historical active power output samples} for each power plant are collected by selecting corresponding stations in the Zhoushan islands in China from 2015 to 2018. 
\subsubsection{Accuracy of NN-based Surrogate Model}
We first investigate the suitability of employing a \textcolor{black}{NN-based} model as a surrogate for the VSI. PLS is unnecessary since the scale of both two cases is manageable. \textcolor{black}{A set of 10,000 $\{\boldsymbol{z},\sigma\}$ samples is generated following the procedures outlined in Section II.B. This dataset is subsequently divided into a training set consisting of 7,000 samples and a test set comprising 3,000 samples. Then the approximate analytical expression $\sigma(\boldsymbol{z})$ is obtained by training the NN with these simulated samples. To further assess the accuracy and the generalization capability of the NN model, additional 10000 $\{\boldsymbol{z},\sigma\}$ samples are randomly generated based on the same procedure and are appended to the test set.} 

\textcolor{black}{By feeding both training and test sets into the obtained NN model, we extract predicted VSI values and quantify the associated prediction errors. The accuracy test results are displayed in Fig. \ref{fig:2}. The first row depicts histograms of prediction errors in the two test systems. In the 14-bus system, the mean absolute errors for the training and testing sets are $3.18\times10^{-6}$ and $3.29\times10^{-6}$, respectively. Likewise, the corresponding errors are $1.09\times10^{-6}$ and $1.12\times10^{-6}$ in the 30-bus system. The second row compares the actual and predicted values of the VSI. The blue points represent $(\sigma,\sigma^{'})$ pairs, where $\sigma$ and $\sigma^{'}$ are the actual and predicted values of the VSI, respectively. Obviously, most pairs are very close to the red line, which represents the position where predictions equal actual values. These observations indicate that the NN model can serve as an accurate surrogate model for the VSI.} 
\begin{figure}[!t]
\centering
\subfloat[Error histogram]{\includegraphics[width=1.74in]{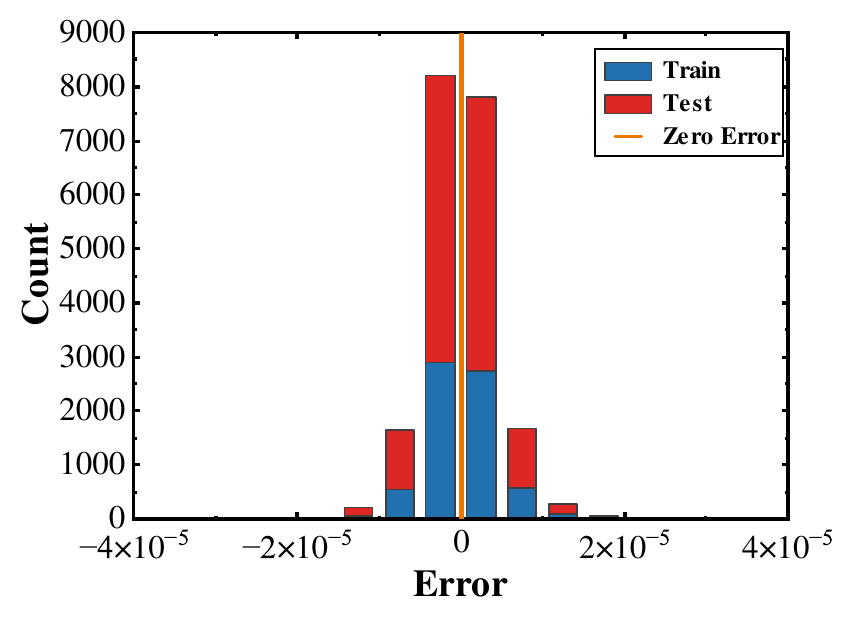}}
\subfloat[Error histogram]{\includegraphics[width=1.74in]{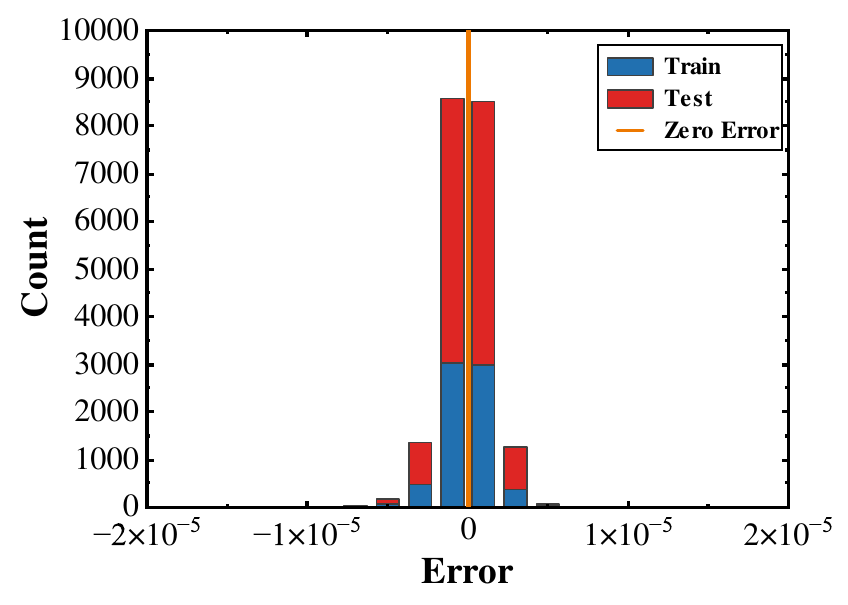}}\\[-1ex]
\subfloat[\textcolor{black}{Prediction vs. actual value}]{\includegraphics[width=1.74in]{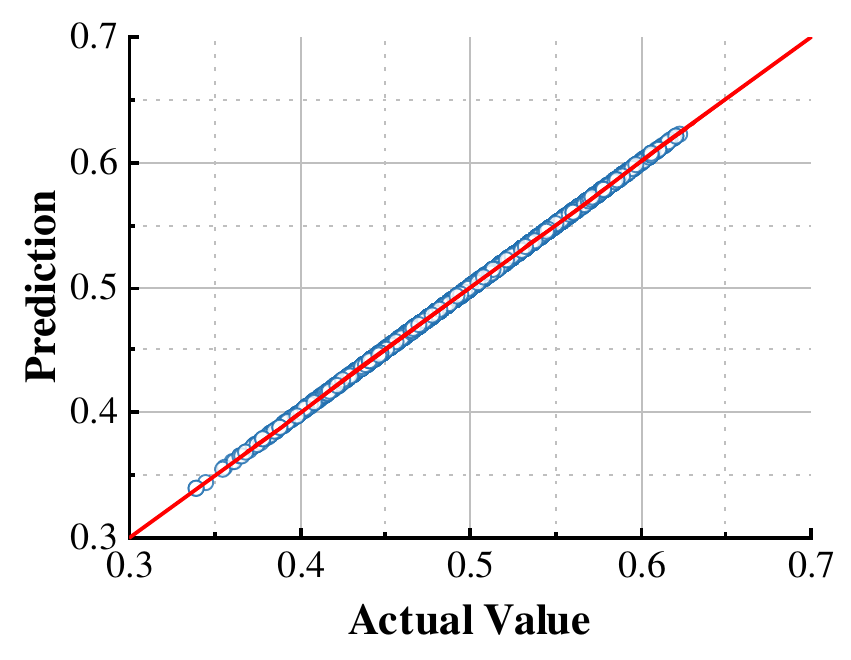}}
\subfloat[\textcolor{black}{Prediction vs. actual value}]{\includegraphics[width=1.74in]{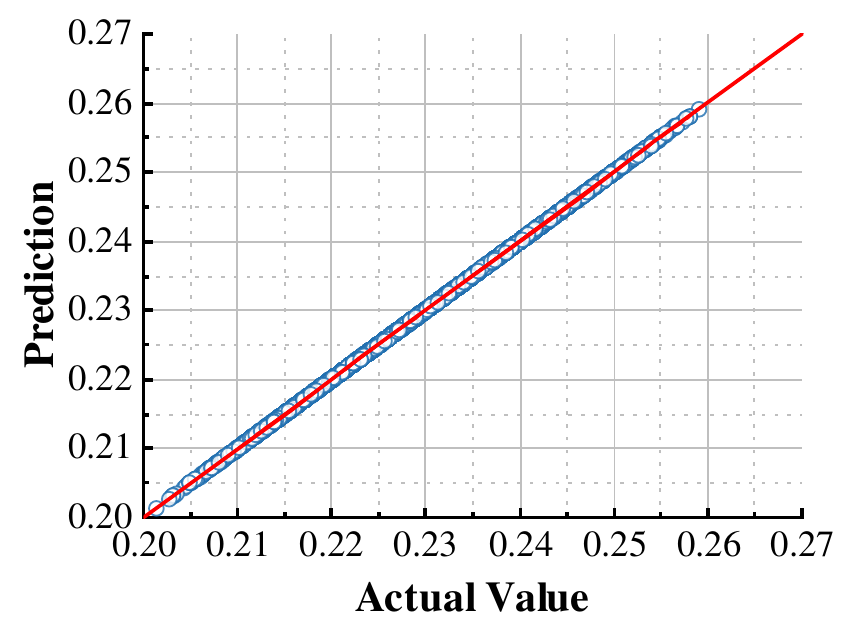}}
\caption{Accuracy of NN-based surrogate model: left column for IEEE 14-bus system, and right column for IEEE 30-bus system}
\label{fig:2}
\vspace{-20pt}
\end{figure}
\subsubsection{Accuracy of Uncertainty Quantification}
Now, we provide empirical validation of the accuracy of the APCE in terms of uncertainty quantification, utilizing MC simulations with 11680 samples as the benchmark. In this case study, a second-order APCE is adopted to represent system responses. Specifically, we select the VSI and the reactive power generation as the quantities of interest (QoIs) for illustration purposes. For each system, we calculate the simulated probability distributions of QoIs under the same operation condition, employing kernel smoothing methods to derive distribution curves. Visual comparisons between the APCE and MC are presented in Fig. \ref{fig:3}. Consistent with the findings in other relevant works \cite{xu2021iterative,muhlpfordt2019chance}, the APCE metamodel of order 2 is sufficient to capture the uncertainty of state variables in power systems. Furthermore, the results suggest that the second-order APCE can effectively model the VSI, a more intricate quantity. In conclusion, the data-driven APCE can represent random system responses accurately without any presumptions or inferences regarding individual distributions and underlying dependence structure.
\begin{figure}[!t]
\centering
\subfloat[]{\includegraphics[width=1.74in]{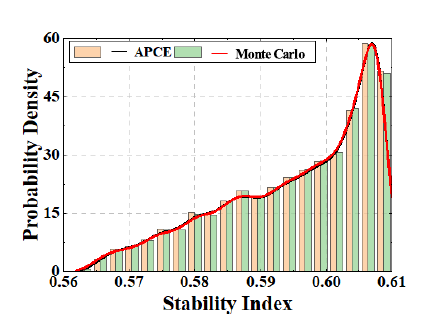}}
\subfloat[]{\includegraphics[width=1.74in]{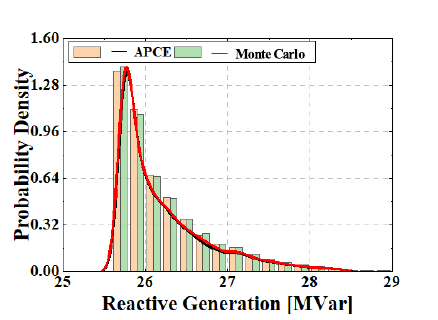}}\\[-1ex]
\subfloat[]{\includegraphics[width=1.74in]{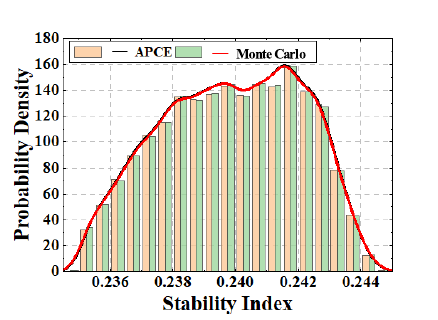}}
\subfloat[]{\includegraphics[width=1.74in]{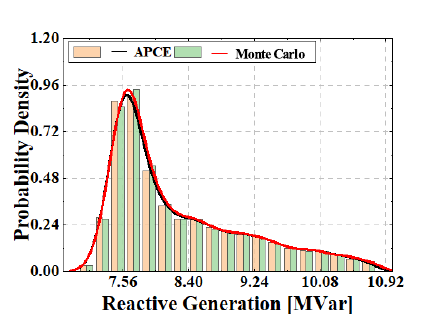}}
\caption{Comparison between MC simulations and APCE metamodel: top row for IEEE 14-bus system, and bottom row for IEEE 30-bus system}
\label{fig:3}
\vspace{-10pt}
\end{figure}
\subsubsection{Satisfaction of Chance Constraints}
Next, the following two CC-VSC-OPF methods are considered to illustrate the effectiveness of the proposed approach: 
(\textbf{M1}) The proposed \textbf{Algorithm 2};
(\textbf{M2}) The identical iterative scheme except that the chance constraints are evaluated via MC simulations in each iteration. \textcolor{black}{MC simulation provides the most accurate characterization of probability distributions, and it has been widely adopted as the benchmark in many research fields to test the performance of statistical-based surrogate models \cite{moustapha2022active,klink2022efficient,peng2021multi}. In this paper, M2 is used to determine the proximity of the results of M1 to the benchmark.}
The performances of \textbf{M1} and \textbf{M2} are tested in terms of the total cost, computational efficiency, and the maximum empirical constraint violation probability $\epsilon_J$ under different settings. The results are summarized in Table \ref{tab:1}. 
As anticipated, decreasing predetermined violation probabilities will entail an increase in the cost to ensure a safer operation. It can be observed from Table \ref{tab:1} that our proposed method, \textbf{M1}, yields solutions that are of comparable quality to that of \textbf{M2} across different settings while exhibiting superior computational efficiency, rendering itself a computationally enhanced algorithm for CC-VSC-OPF problems.   

Based on local testing, the average CPU time required to solve the OPF problem of 14-bus and 30-bus systems is 0.0170s and 0.0357s, respectively. When it comes to the VSC-OPF problem, this time increases to 0.363s and 4.746s. These appreciable differences suggest scalability issues of larger-scale systems, as will be discussed in Section \ref{large}.
\begin{table}[t!]
\centering
\caption{Performances Of M1 and M2 Under Different Settings in IEEE 14-Bus and 30-Bus Systems}
\label{tab:1}
\begin{threeparttable}
\renewcommand{\arraystretch}{1}
\begin{tabular}{cccccc} 
\toprule
\multicolumn{6}{c}{\textbf{Group 1:}\;$\boldsymbol{\epsilon=5\%}$\quad \textbf{14-bus system}}                     \\ 
\midrule
Method & Iterations & Time [s] & OPF Time [s]\tnote{*} & Cost [\$] & $\epsilon_J$ [\%] \\ 
\hline
\textbf{M1}     & 16          & 13.24 & 5.84 & 5378.741  & 5.09                             \\ 
%\cline{5-5}
\textbf{M2}     & 21         & 419.29 & 7.71 & 5378.720  & 5.00                             \\ 
\toprule
\multicolumn{6}{c}{\textbf{Group 2:}\;$\boldsymbol{\epsilon=8\%}$\quad \textbf{14-bus system}}                       \\ 
\midrule
Method & Iterations & Time [s] & OPF Time [s] & Cost [\$] & $\epsilon_J$ [\%]                       \\ 
\hline
\textbf{M1}     & 2          & 1.69 & 0.69 & 5370.110  & 8.00                             \\
\textbf{M2}     & 3          & 60.42 & 1.04 & 5370.084  & 8.00                             \\ 
\toprule
\multicolumn{6}{c}{\textbf{Group 3:}\;$\boldsymbol{\epsilon=5\%}$\quad \textbf{30-bus system}}                     \\ 
\midrule
Method & Iterations & Time [s] & OPF Time [s]\tnote{*} & Cost [\$] & $\epsilon_J$ [\%] \\ 
\hline
\textbf{M1}     & 6          & 32.60 & 28.69 & 356.596  & 5.02                             \\ 
%\cline{5-5}
\textbf{M2}     & 10         & 264.80 & 50.08 & 356.598  & 5.00                             \\ 
\toprule
\multicolumn{6}{c}{\textbf{Group 4:}\;$\boldsymbol{\epsilon=8\%}$\quad \textbf{30-bus system}}                       \\ 
\midrule
Method & Iterations & Time [s] & OPF Time [s] & Cost [\$] & $\epsilon_J$ [\%]                       \\ 
\hline
\textbf{M1}     & 6          & 32.24 & 28.26 & 356.395  & 8.03                             \\
\textbf{M2}     & 12          & 316.84 & 59.44 & 356.394  & 8.00                             \\ 
\midrule
\end{tabular}
\begin{tablenotes}
            \item[*] Sum of OPF solving time in all iterations.
        \end{tablenotes}
\end{threeparttable}
\vspace{-15pt}
\end{table}
\subsubsection{Improvement of Voltage Stability}
This subsection highlights the merits of the proposed CC-VSC-OPF model in optimizing voltage stability. A comparative analysis is conducted by adopting the traditional CC-OPF model without the voltage stability constraint. Both models are solved under the same iterative scheme, with a confidence level of $1-\epsilon=0.95$. The resulting probability distributions of \textcolor{black}{the VSI $\sigma$} are illustrated using the split violin plots. As vividly shown in Fig. \ref{fig:4}, prior to the incorporation of the voltage stability constraint, there exists a certain probability that $\sigma$ fails to meet the requisite threshold, which indicates the system's vulnerability to voltage collapse under specific contingencies. 
With the voltage stability constraint, the stability threshold can be satisfied at the predetermined confidence level, emphasizing the significance of integrating the voltage stability constraint within the optimization framework to avoid potential instability issues.

\begin{figure}
\centering
\includegraphics[width=3.3in]{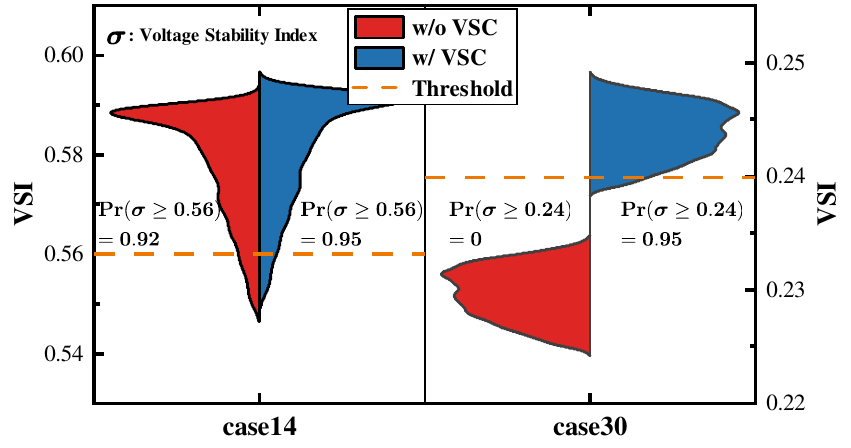}
\caption{\textcolor{black}{Probability distribution of the VSI with and without voltage stability constraint in small cases}}
\label{fig:4}
\end{figure}
\subsubsection{\textcolor{black}{Accuracy of the Tails}} \textcolor{black}{It is also necessary to examine the performance of the APCE method for tail events with much smaller violation probabilities to reflect more stringent reliability criteria, especially for variables ${\boldsymbol{p}, \boldsymbol{q}}$. In this subsection, we will focus on the accuracy of the second-order APCE in approximating quantile values. With the results based on MC simulations as the benchmark, we calculate the root mean square error (RMSE) of quantile values under three different settings of $\epsilon$: $0.1\%, 1\%,$ and $3\%$. The results are presented in Table \ref{tab:4}, where $Q_{\epsilon}^{p/q}$ denotes $\epsilon$-quantile values of corresponding variables.} 

\textcolor{black}{From the test results, we observe that the APCE can consistently deliver an accurate and stable estimation of quantile values. Since the successful application of Algorithm 2 hinges on the accuracy of the APCE approximation, it is anticipated that Algorithm 2 can still perform well even for relatively small $\epsilon$.}
\begin{table}[t!]
\centering
\caption{\textcolor{black}{Accuracy of the APCE for Tails}}
\label{tab:4}
\begin{threeparttable}
\renewcommand{\arraystretch}{1.3}
\begin{tabular}{cccc} 
\toprule
\multicolumn{4}{c}{\textbf{\textcolor{black}{Group 1: 14-bus system}}}                     \\ 
\midrule
\multicolumn{1}{c}{} & \textcolor{black}{$\epsilon=0.1\%$} & \textcolor{black}{$\epsilon=1\%$} & \textcolor{black}{$\epsilon=3\%$} \\ 
\hline
\multicolumn{1}{c|}{\textcolor{black}{RMSE of $Q_{\epsilon}^{p}$ [p.u.]}} & \textcolor{black}{6.74$\times 10^{-5}$} & \textcolor{black}{3.18$\times 10^{-5}$} &  \textcolor{black}{8.24$\times 10^{-6}$}                          \\ 
\multicolumn{1}{c|}{\textcolor{black}{RMSE of $Q_{1-\epsilon}^{p}$ [p.u.]}} & \textcolor{black}{2.00$\times 10^{-4}$} & \textcolor{black}{1.42$\times 10^{-4}$} &  \textcolor{black}{6.39$\times 10^{-5}$}                                                      \\ 
\multicolumn{1}{c|}{\textcolor{black}{RMSE of $Q_{\epsilon}^{q}$ [p.u.]}} & \textcolor{black}{2.41$\times 10^{-4}$} & \textcolor{black}{1.88$\times 10^{-4}$} &  \textcolor{black}{1.04$\times 10^{-4}$}                                                      \\
\multicolumn{1}{c|}{\textcolor{black}{RMSE of $Q_{1-\epsilon}^{q}$ [p.u.]}} & \textcolor{black}{6.10$\times 10^{-4}$} & \textcolor{black}{4.13$\times 10^{-4}$} &  \textcolor{black}{2.35$\times 10^{-5}$}                                                      \\
\toprule
\multicolumn{4}{c}{\textbf{\textcolor{black}{Group 2: 30-bus system}}}                       \\ 
\midrule
\multicolumn{1}{c}{} & \textcolor{black}{$\epsilon=0.1\%$} & \textcolor{black}{$\epsilon=1\%$} & \textcolor{black}{$\epsilon=3\%$} \\ 
\hline
\multicolumn{1}{c|}{\textcolor{black}{RMSE of $Q_{\epsilon}^{p}$ [p.u.]}} & \textcolor{black}{1.78$\times 10^{-4}$} & \textcolor{black}{8.77$\times 10^{-5}$} &  \textcolor{black}{3.75$\times 10^{-5}$}                          \\  
\multicolumn{1}{c|}{\textcolor{black}{RMSE of $Q_{1-\epsilon}^{p}$ [p.u.]}} & \textcolor{black}{7.30$\times 10^{-5}$} & \textcolor{black}{4.70$\times 10^{-5}$} &  \textcolor{black}{1.27$\times 10^{-6}$}                          \\ 
\multicolumn{1}{c|}{\textcolor{black}{RMSE of $Q_{\epsilon}^{q}$ [p.u.]}} & \textcolor{black}{9.86$\times 10^{-5}$} & \textcolor{black}{3.13$\times 10^{-5}$} &  \textcolor{black}{2.71$\times 10^{-5}$}                          \\ 
\multicolumn{1}{c|}{\textcolor{black}{RMSE of $Q_{1-\epsilon}^{p}$ [p.u.]}} & \textcolor{black}{2.15$\times 10^{-4}$} & \textcolor{black}{1.48$\times 10^{-4}$} &  \textcolor{black}{7.14$\times 10^{-5}$}                          \\ 
\midrule
\end{tabular}
\end{threeparttable}
\vspace{-15pt}
\end{table}
\subsection{Real-world test cases}\label{large}
\begin{table*}[t]
    \centering
\caption{Performances Of M1 and M2 With Different PLS Components in 89-Bus and 200-Bus Systems}
\label{tab:2}
\renewcommand{\arraystretch}{1}
\begin{tabular}{cccccccc} 
\toprule
\multicolumn{1}{c|}{}&\multicolumn{4}{c}{89-bus system} & \multicolumn{3}{c}{200-bus system}\\
\multicolumn{1}{c|}{} &\shortstack{\textbf{M1}\\10 Components} &\shortstack{\textbf{M1}\\20 Components} &\shortstack{\textbf{M1}\\w/o PLS} &\multicolumn{1}{c|}{\textbf{M2}} &\shortstack{\textbf{M1}\\25 Components}&\shortstack{\textbf{M1}\\49 Components}&\shortstack{\textbf{M2}\\49 Components} \\
\midrule
\multicolumn{1}{c|}{Iteration} & 8 &8 &8 & \multicolumn{1}{c|}{8} &9&8&8\\
\multicolumn{1}{c|}{Time [s]} & 10.53 &12.20 &895.74 & \multicolumn{1}{c|}{1293.37} &32.21&59.01&1124.15\\
\multicolumn{1}{c|}{OPF Time [s]}&2.49&3.20&887.01&\multicolumn{1}{c|}{880.67} &7.50&38.40&39.72\\
\multicolumn{1}{c|}{$\epsilon_J$ [\%]} &5.02&5.00&5.02&\multicolumn{1}{c|}{5.00}&5.00&5.00&5.00\\
\multicolumn{1}{c|}{RMSE of $p$ [p.u.]}&0.0028&0.0029&8.37$\times10^{-6}$&\multicolumn{1}{c|}{0}&1.80$\times10^{-4}$&6.70$\times10^{-7}$&0 \\
\multicolumn{1}{c|}{RMSE of $v$ [p.u.]}&2.10$\times10^{-5}$&1.97$\times10^{-5}$&2.76$\times10^{-6}$&\multicolumn{1}{c|}{0}&0.0013&1.70$\times10^{-7}$&0 \\
\midrule
\end{tabular}
\vspace{-15pt}
\end{table*}
This case study is conducted on the European 89-bus and Illinois 200-bus systems, with the corresponding voltage stability thresholds being 2.05 and 0.15. Seven 750 MW solar power plants and ten 75 MW solar power plants are added to the above systems, respectively. 
\subsubsection{Effectiveness of PLS-NN Framework}
In this section, we will demonstrate the potential of the PLS-NN framework to overcome scalability issues. The accuracy of the PLS-NN framework is first assessed. To determine the dimension of $\boldsymbol{\upsilon}$ for each system, PLS is initially performed with $2N-1$ components. The percentage variance explained in the response $\sigma$ is plotted against the number of PLS components. As illustrated in Fig. \ref{fig:5}, approximately 99.9\% of the variance in $\sigma$ is explained by the first 20 components in the 89-bus system, while nearly 99.6\% of the variance is accounted for by the first 49 components in the 200-bus system. Accordingly, the dimensions of the latent variables $\boldsymbol{\upsilon}$ are set to 20 and 49 for the two respective systems. 
The performance of resulting reduced surrogate models $\hat{\sigma}(\boldsymbol{\upsilon})$ is presented in Fig. \ref{fig:6}, which demonstrates that the PLS-NN framework is capable of constructing a more parsimonious model while maintaining an acceptable level of accuracy. 
\begin{figure}[!t]
\vspace{-10pt}
\centering
\subfloat[89-bus system]{\includegraphics[width=1.74in]{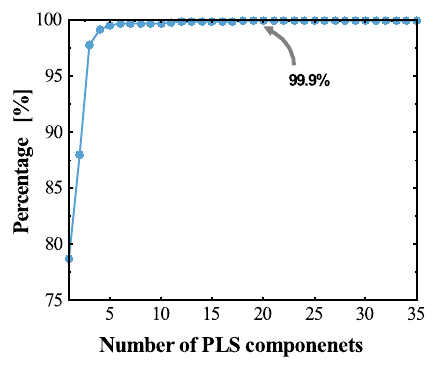}}
\subfloat[200-bus system]{\includegraphics[width=1.74in]{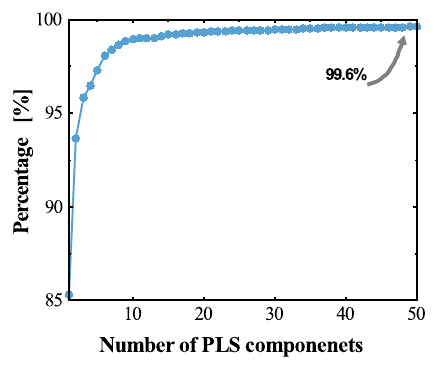}}
\caption{Percent variance explained in the response}
\label{fig:5}
\vspace{-15pt}
\end{figure}
\begin{figure}[!t]
\centering
\subfloat[89-bus system]{\includegraphics[width=1.74in]{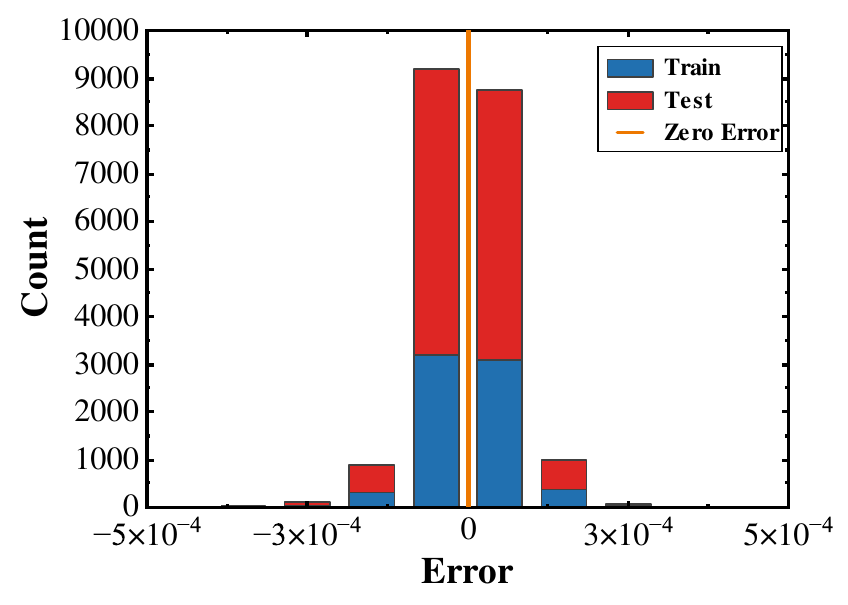}}
\subfloat[200-bus system]{\includegraphics[width=1.74in]{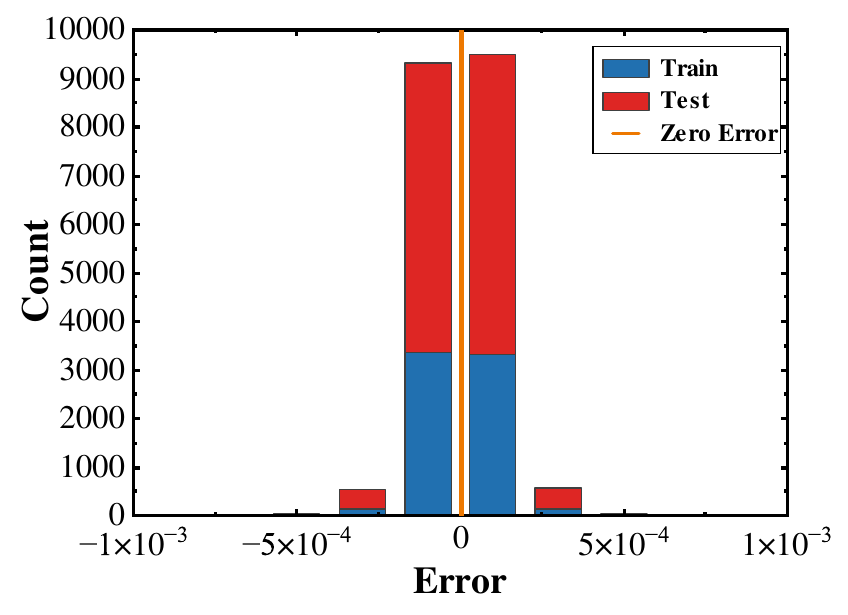}}
\caption{Accuracy of reduced surrogate model}
\label{fig:6}
\vspace{-15pt}
\end{figure}

Next, we aim to elucidate the computational advantages of the proposed framework. Specifically, a comprehensive comparative analysis is conducted, focusing on the solution results under varying numbers of PLS components. To evaluate the performance of our approach, the RMSE of operation points is additionally calculated, using results obtained with the \textbf{M2} method as a benchmark. Note that in the 200-bus system, the deterministic VSC-OPF problem will not be solved within 8h without PLS due to the substantial computational burden involved in evaluating the 399$\times$399 Hessian matrix. Consequently, the benchmark results are calculated with 49 PLS components. As clearly contrasted in Table \ref{tab:2}, the inclusion of PLS significantly enhances computational efficiency. In the 89-bus system, the average CPU time for solving the OPF problem exceeds a hundred seconds without PLS. Conversely, upon incorporating PLS, this time diminishes to a mere fraction of a second. This efficiency gain is achieved by evaluating much smaller Hessian matrices (10$\times$10 or 20$\times$20), as opposed to the larger 177$\times$177 matrix. A comparable improvement is observed in the 200-bus system, wherein the overall CPU time can be reduced to less than one minute. 

\textcolor{black}{When comparing the results of the 89-bus system to those of the 30-bus system, it is revealed that the former requires more iteration steps but has a significantly lower solving time. This seemingly contradictory outcome can be attributed to the utilization of PLS in the 89-bus system. Specifically, the average time for each VSC-OPF solution in the 89-bus system is 0.31s (10 components) or 0.40s (20 components), which is much lower than the approximately 4.7s required for the 30-bus system. The efficiency improvement is credited to the reduced computational burden of calculating a 10 $\times$ 10 or 20 $\times$ 20 Hessian matrix, compared to the 59 $\times$ 59 Hessian matrix in the 30-bus system. As a result, under the PLS-NN framework, the solution time for the 89-bus system is significantly lower than that of the 30-bus system, due to the reduced OPF solving time. This result has once again verified the effectiveness of PLS in the proposed solution framework to mitigate computational complexity.}

Moreover, by examining the RMSE of operation points obtained with different PLS components, it is evident that the solutions achieved with PLS are all in close proximity to the benchmark results. Besides, the distributions of the VSI before and after considering the voltage stability constraint are depicted in Fig. \ref{fig:7}, which verifies the capability of the proposed model to improve voltage stability.
This demonstrates that the proposed algorithm, when combined with PLS, can effectively handle the CC-VSC-OPF problem in large-scale systems while consistently delivering satisfactory results. 
\begin{figure}
\centering
\includegraphics[width=3.3in]{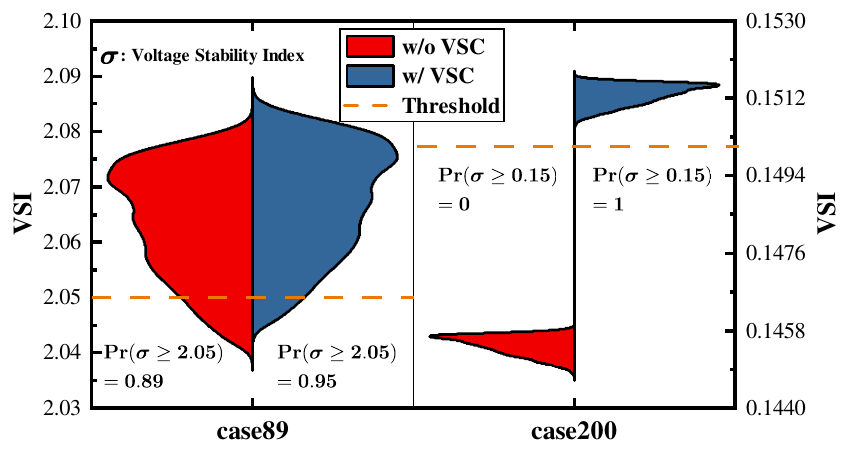}
\caption{\textcolor{black}{Probability distribution of the VSI with and without voltage stability constraint in real-world cases}}
\label{fig:7}
\vspace{-15pt}
\end{figure}
\subsubsection{Comparison of Different Settings of DD-APCE}
As mentioned in Section \ref{DD-APCE}, the efficacy of the proposed APCE diminishes as the dimensionality of uncertainty increases. In this section, we advocate the DD-APCE when faced with high-dimensional random variables. Four different parameter settings are considered in the 200-bus system: (\textbf{S1}) $s=1,m=2$; (\textbf{S2}) $s=2,m=2$; (\textbf{S3}) $s=1,m=3$; (\textbf{S4}) $s=3,m=3$. Fig. \ref{fig:8} illustrates the probability distribution of $\sigma$ obtained under each setting, demonstrating that all four settings accurately characterize uncertainty. This conclusion is further supported by Table \ref{tab:3}, where the RMSE is computed against the benchmark result obtained with the \textbf{M2} method. While higher interaction orders yield more accurate results, they also entail a larger number of basis functions and, consequently, longer evaluation time. In contrast, in this case, settings \textbf{S1} and \textbf{S3} can expedite the determination of expansion coefficients by restricting the interaction order while maintaining sufficient accuracy. 
\begin{figure}[!t]
\centering
\includegraphics[width=3.4in]{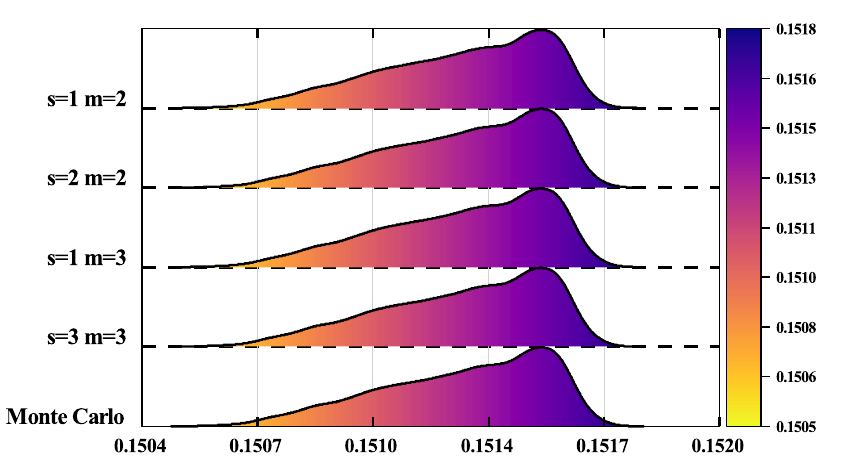}
\caption{Probability distribution of the voltage stability index under different settings of DD-APCE}
\label{fig:8}
\vspace{-15pt}
\end{figure}
\begin{table}[!t]
    \centering
        \caption{Comparison of Different Parameter Settings of DD-APCE in the 200-bus system}
    \label{tab:3}
    \renewcommand{\arraystretch}{1}
    \begin{threeparttable}
    \begin{tabular}{c|cccc}
    \toprule
    &\textbf{S1}&\textbf{S2}&\textbf{S3}&\textbf{S4} \\
    \midrule
    No. of Basis Functions&21&66&31&286\\
    Evaluation Time [s]\tnote{*}&1.23&2.58&1.55&32.14\\
    $\epsilon_J$ [\%]&5.00&5.00&5.00&5.00\\
    RMSE of $p$ [$10^{-5}$p.u.]&1.3&0.067&0.34&0.045\\
    RMSE of $v$ [$10^{-6}$p.u.]&2.7&0.17&6.0&0.067\\
    \midrule
    \end{tabular}
    \begin{tablenotes}
            \item[*] The time to determine APCE coefficients per iteration.
        \end{tablenotes}
\end{threeparttable}
\vspace{-15pt}
\end{table}
The difference in the number of basis functions and evaluation time will become more pronounced as the dimension increases. For instance, when $N=50$, \textbf{S2} involves 1326 basis functions, whereas \textbf{S1} requires only 101.
Therefore, DD-APCE is expected to be markedly more effective in dealing with high-dimensional cases, provided that a lower-variate truncation is adequate for uncertainty propagation.
\subsection{\textcolor{black}{Further Discussions}}
\subsubsection{\textcolor{black}{Stability Guarantee with Approximation Errors}}
\textcolor{black}{
    Although the NN and the truncated APCE metamodels have exhibited high precision for voltage stability index approximation, they will inevitably introduce approximation errors that will potentially undermine voltage stability. Under this circumstance, we further propose a calibration process to accommodate the approximation errors. Specifically, for the error caused by the NN model, the largest empirical prediction error $\rho$ can be obtained after model training. As for the truncation error arising from the APCE, existing literature has derived the upper bound or exact expression \cite{fag2011on,mu2018comment}, which is denoted as $\delta$ here. Hence, the following changes can be applied to Algorithm 2:\\     
    (1) When solving the deterministic VSC-OPF at step 8 and step 18, the voltage stability constraint should be replaced by:
    \begin{equation}
            \sigma(\boldsymbol{z})-\rho\geq\sigma^{\min}.
    \end{equation}
    (2) When evaluating chance constraint satisfaction at step 12, the tightening margin $\Delta \sigma_{\min}^{(k)}$ should be updated as:
    \begin{equation}
        \Delta \sigma_{\min}^{(k)}=\left[\sigma^{\min}-Q_{\epsilon_{\sigma}}^{(k)}\left(\sigma\right)+\delta\right]^+.
    \end{equation}
    Based on the above minor adjustments, voltage stability can be guaranteed in the presence of approximation errors.}
    \subsubsection{\textcolor{black}{Extension to Joint Chance-Constrained (JCC) Problems}}
    \textcolor{black}{The incorporation of joint chance constraints into OPF can provide much stronger guarantees of system security. It is of interest to explore the feasibility of extending the proposed solution framework to address JCC OPF. As can be seen from equation (23), the current framework encounters difficulties when being applied to JCC OPF due to the absence of violation probabilities for individual chance constraints, which precludes the calculation of tightening margins for each variable. Nevertheless, a preprocessing step can render the framework amenable to JCC OPF problems. By leveraging Boole's inequality \cite{boole_2009} or an improved version devised by Baker and Bernstein
 \cite{kyri2019joint}, the joint chance constraint can be reduced into multiple single chance constraints. This transformation enables the subsequent application of the proposed methodology to the transformed single chance-constrained OPF problem, as we have considered in this work. Note that, this transformation can be overly conservative and lead to uneconomic operation strategies.}
%\vspace{-15pt}
\section{Conclusion}\label{conclusion}
In this paper, a CC-VSC-OPF model is proposed to enhance system stability in the presence of arbitrarily distributed and complexly correlated uncertainties. To address the inherent intractability of this model, a two-fold strategy is employed. Firstly, the implicit voltage stability constraint is explicitly embedded by harnessing the NN-based surrogate model. Secondly, to avoid intensive numerical simulations, chance constraints are tackled in a two-step procedure, in which the first step is to \textcolor{black}{quantify} uncertainty at a fixed operation point using the APCE, and the second step is to iteratively update the operation point to meet desired requirements. Moreover, a PLS-NN framework, along with a variant of APCE, namely DD-APCE, is introduced to enhance the versatility of the proposed solution scheme in large-scale systems and high-dimensional uncertainty scenarios. This integration significantly expands the applicability of the scheme with only a minor compromise in accuracy. Through comprehensive case studies, the proposed solution scheme proves to be computationally enhanced and accurate enough for practical applications.

\textcolor{black}{It should be acknowledged that DD-APCE alone may still encounter the curse of dimensionality in cases where a higher-order expansion is necessary to accurately characterize system variables and the input dimension is high. In the future, we wish to enhance the capability of APCE in high-dimensional cases by, for example, applying dimension reduction techniques to input random variables. Meanwhile, we envision our future work to expand the proposed algorithm to JCC problems, which are more practical for engineering applications.} 

\ifCLASSOPTIONcaptionsoff
  \newpage
\fi

\bibliographystyle{IEEEtran}

\bibliography{main}
\end{document}